\documentclass[12pt]{article}

\usepackage{amsmath, amssymb, amsthm}
\usepackage[round]{natbib}
\usepackage{hyperref}
\usepackage{caption}
\usepackage{algorithm}
\usepackage{algpseudocode}
\usepackage{url}
\usepackage{graphicx}
\usepackage{authblk}
\usepackage{tikz}
\usetikzlibrary{decorations.pathmorphing, shapes, arrows.meta, positioning, calc}
\usepackage{enumitem}
\usepackage{booktabs}
\usepackage{subcaption}
\usepackage{geometry}
\graphicspath{{Fig/}}

\onecolumn
\graphicspath{{Fig/}}

\usepackage[doublespacing]{setspace}
\usepackage[fontsize=12pt]{fontsize}

\usepackage{bm, mathtools, dsfont}
\usepackage{booktabs}

\newcommand{\R}{\mathbb{R}}

\newcommand{\PPo}{\mathbb{P}_{\bm{\omega}}}
\newcommand{\EE}{\mathbb{E}}

\newcommand{\Ehat}{\widehat{E}}

\newcommand{\bhat}{\widehat{\beta}}

\newcommand{\Xbar}{\overline{X}}
\newcommand{\ybar}{\overline{y}}
\newcommand{\neff}{n_{\mathrm{eff}}}
\newcommand{\Om}{\Omega}

\newcommand{\diag}{\operatorname{diag}}

\newcommand{\norm}[1]{\lVert #1 \rVert}

\DeclareMathOperator*{\argmin}{arg\,min}

\makeatletter
\renewcommand\paragraph{\@startsection{paragraph}{4}{\z@}%
  {3.25ex \@plus1ex \@minus.2ex}%
  {-1em}%
  {\normalfont\normalsize\bfseries}}
\makeatother

\usetikzlibrary{decorations.pathmorphing,shapes.geometric,arrows.meta,calc}

\definecolor{inkblue}{HTML}{2E2E8F}
\definecolor{cubefill}{HTML}{C9C9EF}
\definecolor{ptgray}{HTML}{6E6E6E}

\tikzset{
  hand/.style={decorate,decoration={random steps,segment length=9pt,amplitude=.35pt}},
  axis/.style={hand,line width=.9pt,black},
  curve/.style={line width=1.1pt,inkblue,smooth,samples=150,variable=\x},
  globalfit/.style={dashed,dash pattern=on 5pt off 4pt,line width=.9pt,black},
  localfit/.style={dashed,dash pattern=on 5pt off 4pt,line width=1.1pt,inkblue},
  guide/.style={dashed,dash pattern=on 6pt off 5pt,line width=.7pt,black!70},
  pt/.style={circle,fill=ptgray,inner sep=0pt,minimum size=3pt},
  otimes/.style={circle,draw,line width=.9pt,fill=white,inner sep=0pt,minimum size=6.5pt},
  panelcap/.style={font=\small\itshape,align=center,anchor=north},
  midcap/.style={text width=6.3cm,align=left,anchor=north west,font=\small,
            execute at begin node={\hyphenpenalty=10000\exhyphenpenalty=10000}},
}

\newcommand{\wpt}[3]{%
  \node[circle,fill=inkblue,inner sep=0pt,minimum size=2.7pt,opacity=#3] at (#1,#2) {};}

\newcommand{\xmark}[4]{%
  \begin{scope}[opacity=#4]
    \node[otimes,draw=#1] (m) at (#2,#3) {};
    \draw[#1,line width=.8pt] ($(m.center)+(45:3.2pt)$)--($(m.center)+(225:3.2pt)$)
                              ($(m.center)+(135:3.2pt)$)--($(m.center)+(315:3.2pt)$);
  \end{scope}}

\newcommand{\bbcurve}{plot[domain=0.55:6.0]({\x},{1.05+0.26*(\x-1.6)^2-0.020*(\x-1.6)^3-0.05*(\x-0.55)*(6.0-\x)})}

\newcommand{\blackbox}[1]{%
  \begin{scope}[scale=#1,every path/.append style={hand,line width=1pt,inkblue}]
    \fill[cubefill,opacity=.75,draw=inkblue] (0,0)--(2,0)--(2,2)--(0,2)--cycle;
    \fill[cubefill,opacity=.55,draw=inkblue] (0,2)--(.62,2.55)--(2.62,2.55)--(2,2)--cycle;
    \fill[cubefill,opacity=.40,draw=inkblue] (2,0)--(2.62,.55)--(2.62,2.55)--(2,2)--cycle;
  \end{scope}
  \begin{scope}[scale=#1]
    \node[inkblue,font=\Large] at (1,1) {$\hat\phi$};
  \end{scope}}

\newcommand{\axesframe}{%
  \draw[axis] (0.15,-0.15) -- (0.15,4.75);
  \draw[axis] (-0.25,0.20) -- (6.45,0.20);
  \node[font=\Large] at (-0.15,4.55) {$y$};
  \node[font=\Large] at (6.70,-0.05) {$x$};
  \draw[guide] (3.5,-0.60) -- (3.5,4.60);
  \node[font=\large] at (3.5,-0.85) {$x^{*}$};}

\newcommand{\panelL}{%
\begin{scope}[x=1cm,y=1cm]
  \axesframe
  \draw[curve] \bbcurve;
  \draw[globalfit] (0.45,0.287) -- (6.15,3.880);
  \node[pt] at (0.739,1.085) {};
  \node[pt] at (0.782,0.637) {};
  \node[pt] at (0.801,1.389) {};
  \node[pt] at (0.975,0.814) {};
  \node[pt] at (1.020,1.162) {};
  \node[pt] at (1.325,0.704) {};
  \node[pt] at (1.455,1.198) {};
  \node[pt] at (1.489,0.905) {};
  \node[pt] at (1.712,1.033) {};
  \node[pt] at (1.724,0.526) {};
  \node[pt] at (1.955,0.291) {};
  \node[pt] at (1.980,0.614) {};
  \node[pt] at (2.005,1.235) {};
  \node[pt] at (2.095,0.860) {};
  \node[pt] at (2.271,0.583) {};
  \node[pt] at (2.305,1.029) {};
  \node[pt] at (2.587,1.096) {};
  \node[pt] at (2.725,1.397) {};
  \node[pt] at (2.849,1.098) {};
  \node[pt] at (3.064,1.627) {};
  \node[pt] at (3.102,1.404) {};
  \node[pt] at (3.189,1.395) {};
  \node[pt] at (3.271,1.573) {};
  \node[pt] at (3.756,1.785) {};
  \node[pt] at (3.808,1.879) {};
  \node[pt] at (3.830,2.190) {};
  \node[pt] at (4.123,2.052) {};
  \node[pt] at (4.130,2.492) {};
  \node[pt] at (4.231,2.539) {};
  \node[pt] at (4.385,1.802) {};
  \node[pt] at (4.531,2.538) {};
  \node[pt] at (4.569,2.211) {};
  \node[pt] at (4.616,2.948) {};
  \node[pt] at (4.787,2.754) {};
  \node[pt] at (4.851,2.330) {};
  \node[pt] at (4.938,3.120) {};
  \node[pt] at (5.076,3.438) {};
  \node[pt] at (5.186,3.800) {};
  \node[pt] at (5.204,3.054) {};
  \node[pt] at (5.353,3.474) {};
  \node[pt] at (5.384,4.048) {};
  \node[pt] at (5.523,3.750) {};
  \node[pt] at (5.792,3.755) {};
  \node[pt] at (5.881,4.290) {};
  \xmark{black}{3.5}{2.209}{1}
  \xmark{inkblue}{3.5}{1.482}{1}
  \fill[white,rounded corners=2pt] (0.42,3.52) rectangle (3.95,4.92);
  \begin{scope}[shift={(0.55,3.70)}]  \blackbox{0.42} \end{scope}
  \node[font=\scriptsize,align=left,anchor=north west] at (1.72,4.35)
    {fitted by \textbf{any} model:\\XGBoost, TabPFN, \ldots};
  \draw[inkblue,line width=.9pt,-{Triangle[length=12pt,width=8pt]}]
        (1.35,3.50) to[out=-70,in=118] (2.78,1.041);
  \node[panelcap] at (3.1,-1.45) {(a) A global linear fit underfits;\\a flexible black-box model can be more accurate.};
\end{scope}}

\newcommand{\panelR}{%
\begin{scope}[x=1cm,y=1cm]
  \axesframe
  \begin{scope}[opacity=.08]
    \draw[curve] \bbcurve;
    \draw[globalfit] (0.45,0.287) -- (6.15,3.880);
  \end{scope}
  \xmark{black}{3.5}{2.209}{.28}
  \draw[localfit] (2.00,0.51) -- (5.00,2.89);
  \wpt{0.739}{1.085}{0.09}
  \wpt{0.782}{0.637}{0.09}
  \wpt{0.801}{1.389}{0.09}
  \wpt{0.975}{0.814}{0.10}
  \wpt{1.020}{1.162}{0.11}
  \wpt{1.325}{0.704}{0.12}
  \wpt{1.455}{1.198}{0.13}
  \wpt{1.489}{0.905}{0.13}
  \wpt{1.712}{1.033}{0.14}
  \wpt{1.724}{0.526}{0.14}
  \wpt{1.955}{0.291}{0.16}
  \wpt{1.980}{0.614}{0.17}
  \wpt{2.005}{1.235}{0.17}
  \wpt{2.095}{0.860}{0.19}
  \wpt{2.271}{0.583}{0.24}
  \wpt{2.305}{1.029}{0.26}
  \wpt{2.587}{1.096}{0.46}
  \wpt{2.725}{1.397}{0.61}
  \wpt{2.849}{1.098}{0.76}
  \wpt{3.064}{1.627}{0.94}
  \wpt{3.102}{1.404}{0.96}
  \wpt{3.189}{1.395}{0.98}
  \wpt{3.271}{1.573}{1.00}
  \wpt{3.756}{1.785}{0.99}
  \wpt{3.808}{1.879}{0.97}
  \wpt{3.830}{2.190}{0.97}
  \wpt{4.123}{2.052}{0.54}
  \wpt{4.130}{2.492}{0.52}
  \wpt{4.231}{2.539}{0.32}
  \wpt{4.385}{1.802}{0.16}
  \wpt{4.531}{2.538}{0.11}
  \wpt{4.569}{2.211}{0.10}
  \wpt{4.616}{2.948}{0.10}
  \wpt{4.787}{2.754}{0.09}
  \wpt{4.851}{2.330}{0.09}
  \wpt{4.938}{3.120}{0.09}
  \wpt{5.076}{3.438}{0.09}
  \wpt{5.186}{3.800}{0.09}
  \wpt{5.204}{3.054}{0.09}
  \wpt{5.353}{3.474}{0.09}
  \wpt{5.384}{4.048}{0.09}
  \wpt{5.523}{3.750}{0.09}
  \wpt{5.792}{3.755}{0.09}
  \wpt{5.881}{4.290}{0.09}
  \xmark{inkblue}{3.5}{1.482}{.28}
  \node[circle,fill=inkblue,inner sep=0pt,minimum size=6.5pt] at (3.5,1.702) {};
  \node[panelcap] at (3.1,-1.45)
    {(b) The \textbf{local distillation} fit is transparent,\\and nearly as accurate as the flexible black-box model.};
\end{scope}}

\theoremstyle{plain}
\newtheorem{theorem}{Theorem}
\newtheorem{proposition}{Proposition}
\newtheorem{lemma}{Lemma}
\newtheorem{corollary}{Corollary}
\theoremstyle{remark}
\newtheorem{remark}{Remark}
\theoremstyle{definition}

\newtheorem{assumption}{Assumption}

\title{Interpretable AI with local distillation}
\author[1]{Erin Craig\thanks{Corresponding author: \href{mailto:ercr@umich.edu}{ercr@umich.edu}}}
\author[2]{Yiling Huang}
\author[2]{Snigdha Panigrahi}
\affil[1]{Department of Biostatistics, University of Michigan, Ann Arbor, MI, USA}
\affil[2]{Department of Statistics, University of Michigan, Ann Arbor, MI, USA}

\begin{document}

\maketitle
\renewcommand{\baselinestretch}{1.3}
\normalsize

\begin{abstract}
Modern AI models such as tabular foundation models and gradient-boosted ensembles can outpredict classical methods, but provide little basis for reasoning about their predictions. High-stakes decisions call for models that are accurate and transparent as built. Local linear modeling offers a path forward: a smooth regression surface is locally well approximated by a linear one, allowing a linear fit near each query point to achieve high accuracy without sacrificing transparency; the challenges are in defining ``local'' and interpreting the results. Here, we propose \textit{local distillation}, in which a  flexible black-box ``teacher'' model guides the fitting of a regularized linear ``student'' model at each query point. For interpretation, we add a small amount of Gaussian randomization to the local objective and use refits to assess stability. Under the lasso penalty, we prove that the feature-selection probabilities are stable under small perturbations of the training responses. Across 17 benchmark datasets, local distillation nearly matches its AI teacher's accuracy. In a high-dimensional gene expression example, local distillation identifies patient subgroups whose local models use different genes; this heterogeneity is invisible to a global linear model, and difficult to surface in a flexible black-box model.
\end{abstract}


\section{Introduction}
\label{sec:intro}

Suppose we observe features $\bm{X}$ and response $\bm{y}$, and seek to predict $y^*$ for a new observation $\bm{x}^*$ using a model that is sufficiently transparent for decision-making.
A popular choice is the linear model, which provides interpretable predictions and is computationally simple to fit.
But a single linear model fit to the full training set may lack the flexibility needed for accurate prediction; in this case, modern AI methods like tabular foundation models~\citep{hollmann2025accurate} and gradient boosting ensembles can predict better~\citep{kong2026tabfm, qu2025tabicl}.
Such accuracy, however, comes at the cost of interpretability: these flexible black-box models are opaque and provide little basis for reasoning about their predictions, limiting their utility for decision-making.

The standard approach is to explain black-box predictions post hoc, with tools such as LIME~\citep{ribeiro2016lime} and SHAP~\citep{lundberg2017unified}.
These explanations are constructed separately from the model, targeting the black box's output rather than the data, and can be unstable under choices external to the model and data~\citep{garreau2020explaining, aas2021explaining}.
We instead take the position that interpretability should come from the predictive model itself: predictions should be transparent by construction, so that they can be reasoned about directly rather than explained post hoc~\citep{rudin2019stop}.

Here, we propose \textit{local distillation}, a tuning-free method for fitting a local linear model at each new query point $\bm{x}^*$.
Our goal in fitting such a model is two-fold: (i) to match the predictive accuracy of a flexible, potentially non-linear fit, and (ii) to provide interpretable predictions that reveal feature--outcome relationships among similar observations and characterize heterogeneity in the data.
The motivation is simple: a smooth regression function can be well approximated locally by a linear function, whose coefficients readily lend themselves to statistical interpretation.

The key challenge then becomes how to define ``local'': which observations should inform the fit at a given query point? 
We next give an informal overview of our modeling approach and illustrate it with a simple data example that highlights both its predictive performance and the interpretability it provides.

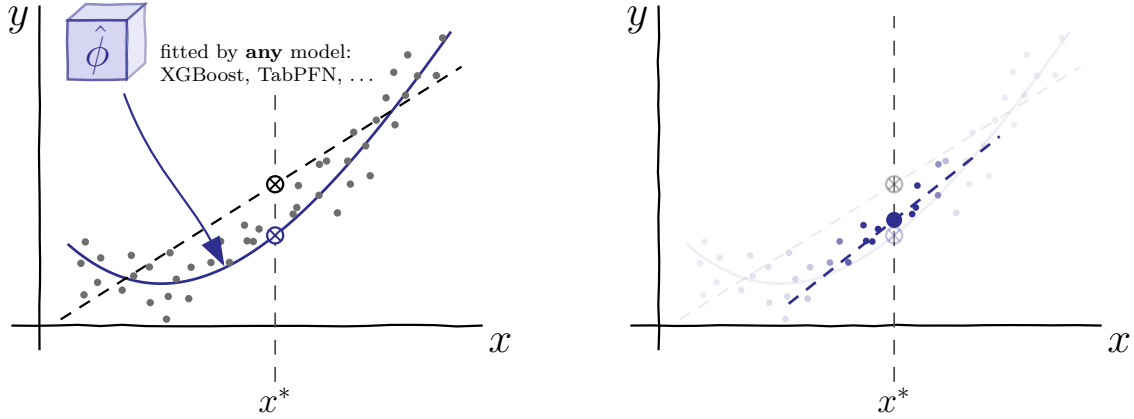
\begin{figure}[!t]
  \centering
  \resizebox{\textwidth}{!}{%
    \pgfmathsetseed{1}
    \begin{tikzpicture}
      \panelL
      \begin{scope}[shift={(8.8,0)}] \panelR \end{scope}
    \end{tikzpicture}}
  \caption{ \textbf{Local distillation} fits a local linear model at $\bm{x}^{*}$, guided by \emph{any} flexible black-box teacher model~$\hat\phi$. 
  The teacher does two jobs: it \textbf{(1)} weights training points by similarity to $\hat\phi(\bm{x}^{*})$, and \textbf{(2)} pulls the fit toward $(\bm{x}^{*},\hat\phi(\bm{x}^{*}))$.} 
  \label{fig:localdist}
\end{figure}

Figure \ref{fig:localdist} illustrates how local distillation works. At any query point, local distillation fits a local linear \textit{student model} guided by a flexible black-box \textit{teacher model} $\hat\phi$, which can be any high-performing model (e.g., a tabular foundation model or a gradient-boosted ensemble). The teacher supplies two ingredients for fitting the student at $\bm{x}^*$:
\begin{enumerate}
    \item \textbf{Training weights:} training observations $\bm{X}$ are weighted according to similarity of teacher predictions $\hat{\phi}(\bm{X})$ to $\hat{\phi}(\bm{x}^*)$. 
    Unlike feature-space locality, which suffers from the curse of dimensionality, locality in the predicted outcome is one-dimensional and provides a natural axis for heterogeneity, as feature--outcome relationships often vary across outcome levels. In medicine, for example, a patient's local model will be fit using patients with similar predicted risk.
    \item \textbf{Prediction anchor:} the teacher's prediction enters the fit as a pseudo-observation, with a data-driven weight that adaptively pulls the student's prediction toward $(\bm{x}^*, \hat\phi(\bm{x}^*))$.
\end{enumerate} 
The final prediction is produced by the fitted linear student model alone.
Additionally, local distillation is gated: if the teacher does not outperform the global linear model, the method reverts to the global linear fit.

For decision-makers, the goal is not only accurate prediction, but also interpretable and trustworthy conclusions, with trust closely tied to their \textit{stability} under small perturbations of the training data~\citep{yu2020veridical, rewolinski2025pcs}.
We therefore add a small amount of Gaussian randomization to the local distillation objective and assess interpretation through stability across randomized refits. These refits identify locally important features without breaking feature dependence or recomputing teacher predictions and, when clustered across query points, reveal stable subgroups in how the feature--outcome relationships varies. Our choice of the randomization scheme is supported by stability guarantees for feature--selection probabilities, and it preserves fidelity to the notion of locality learned from the training data.

\begin{figure}[!t]
  \centering
  \includegraphics[width=\linewidth]{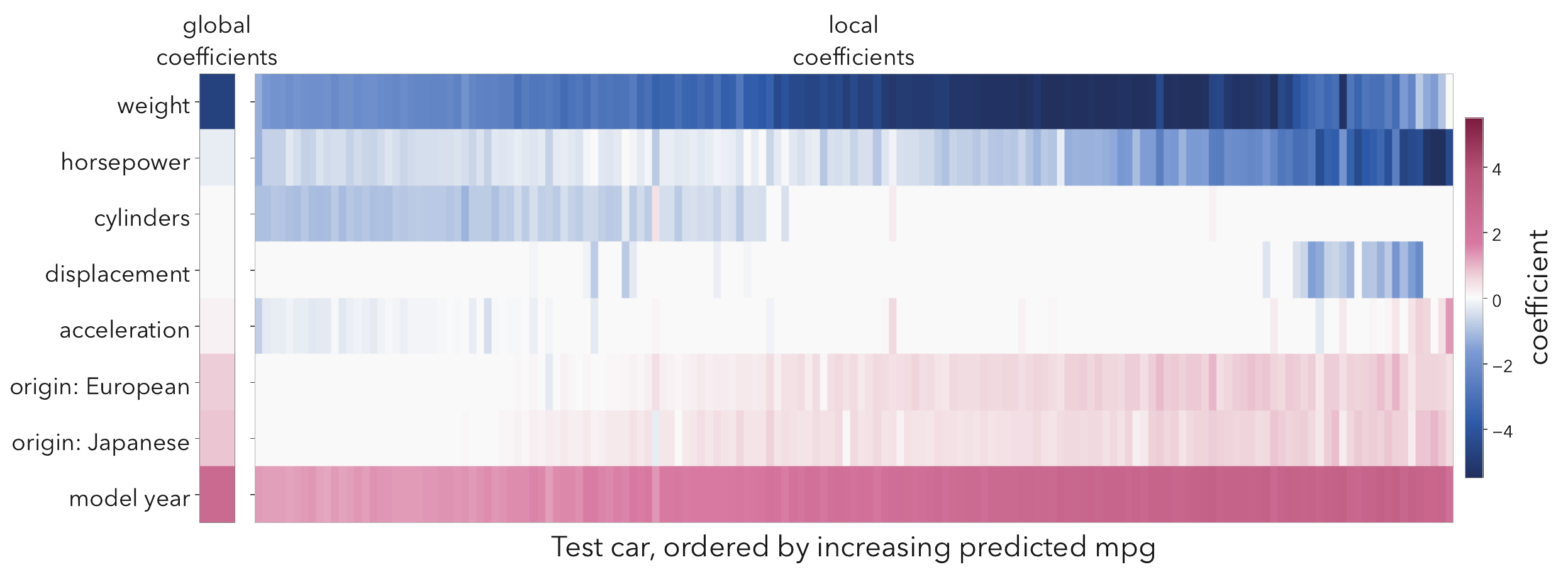}
  \caption{\textbf{Car-specific miles-per-gallon models from local distillation.}
    Fitted coefficients for the Auto MPG data using local distillation with an AI teacher (TabPFN) and a lasso-regularized student. Cars are ordered left to right by TabPFN's predicted MPG; coefficients are per feature standard deviation. The global lasso coefficients are on the left. Local distillation improves over the global lasso's test PSE by 48\% (global lasso regression PSE 10.81; TabPFN 5.25; local distillation 5.59).}
  \label{fig:intro_example}
\end{figure}

As a running example, we predict fuel economy in the Auto MPG data from eight vehicle features~\citep{autompg1993}, using a tabular foundation model (TabPFN, \citealp{hollmann2025accurate}) as the teacher.
Local distillation fits a lasso for each test car, yielding car-specific coefficient vectors shown in Figure~\ref{fig:intro_example}. 
The local models have about half the prediction squared error (PSE) of the global lasso and nearly match the teacher (PSE 5.59; global lasso 10.81; TabPFN 5.25). 
Their coefficients reveal heterogeneity invisible to the global fit:  the cylinder count predicts fuel economy among the least efficient cars but not the most efficient, and displacement is likewise useful only for a subset, while the global lasso sets both to zero.

This work contributes (1) local distillation, a predictive method that uses a black-box teacher to construct accurate, sparse local linear fits \textit{(described in Section~\ref{sec:methods} and demonstrated on 17 benchmark datasets in Section~\ref{sec:real_data})},  and (2) a randomized stability framework for interpreting these fits, both at individual query points and in aggregate, with theoretical guarantees \textit{(Sections~\ref{sec:interpretation} and~\ref{sec:theory})}. 
In two case studies, the framework reveals heterogeneity in the feature--outcome relationship that a global linear model cannot express \textit{(Sections~\ref{sec:interpretation_aggregate} and~\ref{sec:bctcga})}, which is a central but elusive goal of prediction methods, particularly in personalized medicine.

\section{Local distillation}
\label{sec:methods}

Suppose we have training data $\bm{X}\in \mathbb{R}^{n\times p}$ (standardized), a continuous response $\bm{y} \in \mathbb{R}^n$, and a test point $\bm{x}^*$.
We also have a \textit{teacher} $\hat\phi:\mathbb{R}^p \to \mathbb{R}$, any fitted predictor of $y$ from $\bm{x}$; it may be a black box such as TabPFN or XGBoost, and it need not be interpretable.
The \textit{student} is an elastic-net regularized linear model.
Our goal is to predict $y^*$ at $\bm{x}^*$ with a student fit locally (by weighting the training data): the prediction is $\hat y^* = \bm{x}^{*\top}\hat{\bm\beta}^*$ with the teacher guiding the fit of $\hat{\bm\beta}^*$.

We describe local distillation using squared-error loss in Section~\ref{sec:localdistillation} and illustrate it with an example.
The method is general, and Section~\ref{sec:generalization} covers its extension to other loss functions and forms of regularization.

\subsection{Model fitting and prediction}
\label{sec:localdistillation}

Our work is inspired by \textit{knowledge distillation}~\citep{hinton2015distilling}, in which a simpler ``student'' model is trained under the guidance of a more complex ``teacher''.
The student is usually a smaller neural network, and the goal is compression: a fast, lightweight model to approximate a costly one. 
Applying the same principle to regression with a linear student would yield 
\begin{equation}
\hat{\bm\beta} = \argmin_{\bm\beta} \; \frac{1}{2 n} \left[\sum_{j=1}^{n} \bigl(y_j - \bm{x}_j^\top \bm\beta\bigr)^2 + \mu \sum_{j=1}^{n} \bigl(\hat\phi(\bm{x}_j) - \bm{x}_j^\top \bm\beta\bigr)^2\right],
\label{eq:general}
\end{equation}
where the hyperparameter $\mu$ determines the influence of the teacher. 

When $\EE[y \mid \bm{x}]$ is highly nonlinear, however, a single linear student cannot match a flexible teacher's predictive performance.
We therefore introduce local distillation, summarized in Algorithm~\ref{alg:local_distill}, in which a separate student model is fit for each test observation. 
The teacher guides the definition of the local neighborhood around the query point and anchors the fit through its prediction at that point.

More precisely, our proposed method modifies the standard knowledge distillation objective in \eqref{eq:general} in three ways:
\begin{enumerate}
\item We fit a separate \emph{local} model for each test observation $\bm{x}^*$, replacing the global distillation sum with a single pseudo-observation at $\bm{x}^*$ that pulls the student's prediction ${\bm x^*}^\top \bm \beta$ toward the teacher's prediction; see \eqref{eq:loss}. This anchors the local fit at the query point.
\item We replace uniform training weights with similarity weights $\hat{S}_j(\bm{x}^*)$ that upweight training points whose teacher prediction is close to $\hat\phi(\bm{x}^*)$. These weights define \emph{locality}: which training observations inform the fit at $\bm{x}^*$.
\item We estimate $\mu$ as the cross-validated student-to-teacher loss ratio, rather than tuning it (see \eqref{eqn:muhat}), and scale it inversely to $\sqrt{\hat n_{\text{eff}}}$, the effective sample size defined in \eqref{eq:weights}. 
The teacher's influence therefore grows both when the teacher outperforms the student globally (larger $\hat\mu$) and when the local neighborhood is sparse (smaller $\hat n_{\text{eff}}$), where the local fit most needs anchoring. 
When $\hat\mu \le 1$, we return the global linear fit for all test observations.\footnote{This choice is deliberately conservative: a teacher that offers no global improvement may still help locally, but estimating \textit{where} requires localized loss ratios, which we found too variable to be reliable.}
\end{enumerate}

In an ablation study across our benchmark datasets (Appendix~\ref{app:ablation}), we find that using both the similarity weights \textit{and} the teacher prediction anchor yields better predictions than using either alone.

In Algorithm~\ref{alg:local_distill} and throughout, we write $\hat\phi^{(-j)}(\bm{x}_j)$ for the teacher's prediction at training point $j$ computed without access to $(\bm{x}_j, y_j)$: if the teacher is fit to the training data, or uses it as context at prediction time (as with in-context learning), $\hat\phi^{(-j)}$ is its out-of-fold (OOF) prediction; if the teacher makes no use of the training data, $\hat\phi^{(-j)} = \hat\phi$. We additionally include an unpenalized intercept, which we suppress in the notation for clarity.

\begin{algorithm}[!t]
\caption{Local distillation of a teacher $\hat\phi$ into a sparse linear student at $\bm{x}^*$}  
  \textbf{Input:} Training data $(\bm{X}, \bm{y})$ with $n$ observations; test observation $\bm{x}^*$.\\Student: elastic net with mixing parameter $\alpha \in [0,1]$.\\Teacher: any fitted predictor $\hat\phi$.\\
  \textbf{Output:} Local student coefficients $\hat{\bm \beta}^*$ at $\bm{x}^*$, and prediction $\hat{y}^* = \bm{x}^{*\top}\hat{\bm\beta}^*$.
  \hrule
  \begin{enumerate}
    \item \textbf{Estimate distillation strength,} the cross-validated student/teacher loss ratio:
    \begin{equation}
    \hat\mu = \frac{\hat{L}(f_{\hat{\bm\beta}})}{\hat{L}(\hat\phi)},
    \qquad
    \hat{L}(g) = \frac{1}{n}\sum_{j=1}^{n}\left(y_j - g^{(-j)}(\bm{x}_j)\right)^2,
    \label{eqn:muhat}
    \end{equation}
    where $f_{\hat{\bm\beta}}$ is the \emph{global student}, the elastic-net fit to $(\bm{X}, \bm{y})$ with shrinkage $\hat\lambda$ chosen by cross-validation. 
    If $\hat\mu \le 1$, return $\hat{y}^* = f_{\hat{\bm\beta}}(\bm{x}^*)$ and $\hat{\bm\beta}^* = \hat{\bm\beta}$, and stop.
    \item \textbf{Compute similarity weights and effective sample size:}
    \begin{equation}
    \begin{gathered}
      \hat{S}_j = \hat{S}_j(\bm{x}^*) = \frac{\exp(-d_j)}{\sum_{k=1}^n \exp(-d_k)},
      \text{ where } 
      d_j = \frac{\bigl(\hat\phi^{(-j)}(\bm{x}_j) - \hat\phi(\bm{x}^*)\bigr)^2}{\hat\sigma^2_\phi}, \\
      \hat n_{\text{eff}} = \biggl(\sum_{j=1}^n \hat{S}_j^2\biggr)^{-1},
    \end{gathered}
    \label{eq:weights}
    \end{equation}
    where $\hat\sigma^2_\phi$ is the empirical variance of $\{\hat\phi^{(-j)}(\bm{x}_j)\}_{j=1}^n$.  
    \item \textbf{Fit the local model and predict:}
    $\hat{y}^* = \bm{x}^{*\top}\hat{\bm\beta}^*$, where
    \begin{equation}
    \begin{aligned}
      \hat{\bm\beta}^* = \argmin_{\bm\beta} \;
        & \frac{1}{2}  \sum_{j=1}^{n} \hat{S}_j\bigl(y_j - \bm{x}_j^\top\bm\beta\bigr)^2 + \frac{\hat\mu}{2 \sqrt{\hat n_{\text{eff}}}}
              \bigl(\hat\phi(\bm{x}^*) - \bm{x}^{*\top}\bm\beta\bigr)^2 \\
        & + \hat\lambda\bigl[\alpha\|\bm\beta\|_1 + \tfrac{(1-\alpha)}{2}\|\bm\beta\|_2^2\bigr].
    \end{aligned}
    \label{eq:loss}
    \end{equation}
  \end{enumerate}
  \vspace{0.25em}
  {\footnotesize \textit{Note:} For a set of test observations, run step~1 once and steps~2--3 independently for each $\bm{x}^*$.}
  \label{alg:local_distill}
\end{algorithm}

\noindent\textbf{The teacher as a Bayesian prior.} The anchor in \eqref{eq:loss} has a Bayesian interpretation: the teacher's prediction at the query point acts as a Gaussian predictive prior on $\bm{x}^{*\top}\bm{\beta}$, centered at $\hat\phi(\bm{x}^*)$ with precision $\dfrac{\hat\mu}{\sqrt{\hat n_{\text{eff}}}}$. This is a data augmentation prior~\citep{kadane1980interactive, bedrick1996new}, in which beliefs are encoded as pseudo-observations rather than parameter distributions, and the data-driven precision parallels the power prior of \citet{ibrahim2000power}.
Local distillation differs in that each query point gets its own local model with a single pseudo-observation at $\bm{x}^*$ centered at the teacher's prediction, with precision estimated from the data in the spirit of empirical Bayes.

\noindent\textbf{A worked example of local distillation.} On the Auto MPG data of Section~\ref{sec:intro} (60/40 train/test split, TabPFN teacher, lasso student), the global lasso has CV PSE $12.81$ and TabPFN $6.31$, so $\hat\mu = 2.03$; the effective sample sizes $\hat n_{\text{eff}}$ range from 39 to 172 (median 141) out of $n_{\text{train}} = 235$. We fit one model per car by solving \eqref{eq:loss}. On this held-out set, the global lasso has PSE $10.81$, TabPFN $5.25$, and local distillation $5.59$.  In Section~\ref{sec:interpretation}, we visualize and interpret the models.

\subsection{Computability, generalizations and further localization}
\label{sec:generalization}

\noindent\textbf{Computability.} For a single test observation $\bm{x}^*$, Equation~\eqref{eq:loss} is a weighted elastic-net problem with fixed $\hat\lambda$ and $\alpha$: append $\bm{x}^*$ to $\bm{X}$ and $\hat\phi(\bm{x}^*)$ to $\bm{y}$, and use weights $\{\hat S_1, \dots, \hat{S}_n, \hat\mu/\sqrt{\hat n_{\text{eff}}}\}$. Elastic-net solvers that support observation weights, such as \texttt{glmnet}~\citep{friedman2021package} or \texttt{adelie}~\citep{yang2024adelie}, can be used. Many models can be fit in parallel: $\hat\lambda$ is estimated once on the training set, and the local fits are then independent single-$\lambda$ problems.

\noindent\textbf{Other penalties, feature maps, and losses.}
The pseudo-observation construction is agnostic to the penalty: any regularizer supported by the solver, such as the group or fused lasso, can replace the elastic net.
The feature map is likewise a modeling choice: replacing $\bm{x}$ with a basis expansion $\Phi(\bm{x})$ (e.g.\ pairwise interactions) gives a student that is linear in $\Phi(\bm{x})$ and fit identically, at the cost of interpreting coefficients in the expanded basis.
Finally, the method extends to any loss $\ell(y,\eta)$ convex in $\eta = \bm{x}^\top\bm\beta$: the local objective of Algorithm~\ref{alg:local_distill} becomes
\begin{equation}
  \hat{\bm\beta}^* = \argmin_{\bm\beta}
    \sum_{j=1}^{n} \hat{S}_j\,\ell(y_j, \bm{x}_j^\top\bm\beta)
    + \frac{\hat\mu}{\sqrt{\hat n_{\text{eff}}}}\,\ell(\hat\phi(\bm{x}^*), \bm{x}^{*\top}\bm\beta)
    + \hat\lambda[\alpha\|\bm\beta\|_1 + \tfrac{(1-\alpha)}{2}\|\bm\beta\|_2^2],
  \label{eq:general_loss}
\end{equation}
where $\hat\mu$ is the cross-validated student-to-teacher loss ratio with $\ell$ in place of squared error. Taking $\ell(y, \eta) = \tfrac{1}{2}(y-\eta)^2$ recovers Equation~\eqref{eq:loss}.
For logistic regression, the distillation term is the cross-entropy between the teacher's and student's predicted probabilities, matching the soft-target distillation of \citet{hinton2015distilling}; similarity distances $d_j$ are computed on the logit scale. Extending local distillation to survival outcomes with the Cox partial likelihood requires a different anchor construction, as its risk-set structure does not naturally accommodate a pseudo-observation; we leave this to future work.

\noindent\textbf{Further localization.} For each test observation, $\lambda$ could be selected by local cross-validation and $\mu$ estimated from locally weighted losses. We instead reuse the global estimates, since local refitting rarely improved prediction in our experiments.

\noindent\textbf{Cross-modal and cross-domain distillation.} The teacher enters Algorithm~\ref{alg:local_distill} only through its \textit{predictions} (via out-of-fold predictions on the training data, and the prediction at $\bm{x}^*$). There is therefore no requirement that it use the same features as the student. A teacher built on a richer modality (e.g.\ images, text, or additional clinical measurements) can be distilled into a local linear model on the features we wish to \textit{interpret}; conversely, the teacher may perform best with fewer features than the student, as in our gene expression example (Section~\ref{sec:bctcga}).

\section{Related work}
\label{sec:rel_work}

\noindent\textbf{Local linear modeling.}
Local modeling methods fit a separate model for each test observation rather than a single global model. 
The classical example is LOESS~\citep{cleveland1988locally}, which fits a linear or quadratic regression at each query point by weighted least squares, with tricube weights on the nearest fraction of training points in input space and the span tuned by the analyst; distance-based weighting in $\bm{x}$-space degrades quickly as $p$ grows. 
Closer to ours are methods with \textit{supervised} weights derived from a random forest. 
Generalized random forests~\citep{athey2019generalized} and local linear forests~\citep{friedberg2020local} weight by a forest's leaf co-membership; in local linear forests, the local fit is a ridge correction to the forest's own prediction. 
The attention lasso~\citep{craig2025supervised} weights by forest proximity and blends each local fit with a global baseline through a mixing parameter tuned by cross-validation. 
MAPLE~\citep{plumb2018model} uses a forest for feature screening and as a proximity kernel, weighting a local regression on either the observed labels or a black-box model's predictions.

Local distillation is also close to varying coefficient models~\citep{hastie1993varying}, $y = \bm{x}^\top\bm\beta(u)$, with the effect modifier $u = \hat\phi(\bm{x})$ set to the teacher's prediction and $\bm\beta(u)$ estimated by kernel smoothing in $u$~\citep{fan1999statistical} with a sparsity penalty. 
Standard varying coefficient models take the modifier to be a pre-specified covariate; here it is learned by the teacher.

Local distillation differs from these methods in four ways. 
First, locality is defined on the scale of the predicted outcome: two observations are neighbors when the teacher expects them to have similar responses, not when they are close in input space (as in LOESS) or fall in the same leaves (as in the forest-based methods).
This one-dimensional axis does not degrade with $p$, and it is often the axis of scientific interest. 
Second, because the teacher enters only through its predictions, it may be any accurate regressor---a gradient-boosted ensemble, a pretrained foundation model, or a model built on a different feature set entirely. The forest-based methods above are tied to a single model class.
Third, the teacher's influence is estimated rather than tuned: the anchor weight $\hat\mu$ is the cross-validated student-to-teacher loss ratio, and when the teacher offers no improvement, the method reverts to the global fit. 
Finally, the local distillation prediction comes from a single fitted student at each query point, rather than a blend of models or a correction, so each local model has one coefficient vector and one active set, to which the stability guarantees of Section~\ref{sec:theory} apply.

\noindent\textbf{Prediction-powered inference.} Prediction-powered inference~\citep{angelopoulos2023ppi, angelopoulos2023ppipp} and post-prediction inference~\citep{wang2020methods} share intuition with local distillation, though in different settings: all three rely on a powerful ML model to strengthen a classical procedure. PPI's data-estimated weight, which reverts to the classical estimator when the model is not useful, closely parallels our $\hat \mu$; post-prediction inference instead calibrates predicted outcomes against a held-out labeled set before using them in a downstream regression. Both, however, target population-level inference, while our objective is interpretable prediction.

\section{Interpretability through stability}
\label{sec:interpretation}

To support trustworthy interpretation of local distillation, we propose randomized local distillation, a simple modification of Step 3 in Algorithm \ref{alg:local_distill}, that assesses stability by examining how the local coefficients vary under random perturbations of the local distillation optimization. 
This stability analysis supports interpretation at two levels: (i) individually, through the features selected in each sparse local model, and (ii) in aggregate, by characterizing heterogeneity across query points to identify similar and dissimilar regions of the dataset.

\subsection{Randomized local distillation}   
\label{sec:randomized_estimator}

Let $\bm w = (w_1, \dots, w_{n+1})$ be a vector of $(n+1)$ independent and identically distributed randomization variables, with $w_j \sim \mathcal{N}(0, \tau^2)$ for $j = 1, \dots, n+1$.
Rather than relying on a single local fit, we generate repeated randomized local fits by solving
\begin{equation}
\begin{aligned}
  \hat{\bm\beta}^{\bm w} = \argmin_{\bm\beta} \;
   & \frac{1}{2}\sum_{j=1}^{n}  \hat{S}_j\bigl(y_j - \bm{x}_j^\top\bm\beta\bigr)^2
    +  \frac{\hat\mu}{2 (\hat n_{\text{eff}})^{1/2}}
          \bigl(\hat\phi(\bm{x}^*) - \bm{x}^{*\top}\bm\beta\bigr)^2 \\
    &+ \hat\lambda\bigl[\alpha\|\bm\beta\|_1 + \frac{(1-\alpha)}{2}\|\bm\beta\|_2^2\bigr]- \left(\sum_{j=1}^{n} \hat{S}_j^{1/2} w_j  \bm{x}_j   +  \frac{\hat\mu^{1/2}}{n_{\text{eff}}^{1/4}} w_{n+1} \bm{x}^* \right)^\top \bm\beta,
\end{aligned}    
  \label{eq:randomized_objective}
\end{equation}
where $\hat{\bm\beta}^{\bm w}$ denotes the coefficient vector obtained at $\bm{x}^*$ with randomization $\bm{w}$.

Randomized local distillation modifies \eqref{eq:loss} of Algorithm \ref{alg:local_distill} by introducing a linear perturbation through normally distributed noise variables, while leaving the loss and penalty terms unchanged.
Thus, as with local distillation, the randomized fit naturally extends to other loss functions and regularizers.
Analogous to the weighted loss, each randomization variable $w_j$ in the linear perturbation term is scaled by the square root of its corresponding weight, $\hat{S}_j$.
Consequently, if $\hat{S}_j$ is zero or close to zero, $w_j$ has no or little effect on the optimization.

The amount of randomization in each randomized local fit is controlled by $\tau^2$, the variance of the normal randomization variables in $\bm w$.
In practice, we set $\tau^2= t \; *\dfrac{\hat{\sigma}^2}{\hat{n}_{\text{eff}}}$, where $\hat{n}_{\text{eff}} = \frac{1}{\hat{S}_1^2 + \ldots + \hat{S}_n^2}$ is the effective sample size  (Equation~\eqref{eq:weights}), $\hat{\sigma}^2$ is the variance of the teacher’s out-of-fold residuals, and $t$ is a constant chosen by the analyst. 
Choosing the randomization variance as a fraction of $\dfrac{\hat{\sigma}^2}{\hat{n}_{\text{eff}}}$ ensures that the scale of the randomization is comparable to the scale of variability contributed by the data to the optimization objective.
We recommend using data to guide the choice of $t$.
The sensitivity bound established in Theorem~\ref{thm:stability}, from which the stability bound in Corollary \ref{cor:stability} follows, improves as $\tau$ increases, and thus as $t$ increases.
But too much randomization can reduce the predictive accuracy of the randomized fits.
Therefore, we use the largest $t$ whose out-of-fold median prediction error is within a small tolerance of that from the unperturbed fits ($t=0$); the tolerance reflects how much prediction accuracy the analyst is willing to trade for stability, and we use 5\% throughout.
Figure~\ref{fig:single_model} (left) shows the selection process for the Auto MPG data.

{\small{
\begin{remark}
The randomization in \eqref{eq:randomized_objective} is similar in form to that used in the post-selection inference literature, where linear perturbation terms are added to penalized M-estimation problems to preserve information that can later be leveraged for valid inference in selected models.
See, for example, the randomized inference methods developed in \citet{tian2018randomized, panigrahi2021selection, bakshi2024selective, huang2025selective, perry2026post}.
Here, however, the added randomization serves a different role: in Section \ref{sec:theory}, we prove that it provides stability guarantees for feature--selection probabilities under small perturbations of the training data.
\end{remark}
}}

\subsection{Interpreting a single local model}
\label{sec:interpretation_singular}

At a given query point $\bm{x}^*$, the local model characterizes the feature--response relationship among observations with similar predicted outcomes.
More specifically, a single local distillation fit yields a sparse set of features, $\{j:\hat{\beta}_j\neq 0\}$ (for $\alpha > 0$), providing an interpretation of its prediction at $\bm{x}^*$.
But, as with the standard lasso, this selected set may be unstable: predictors near the selection boundary may enter or leave the model under even slight perturbations in the outcomes, ultimately yielding different local interpretations at $\bm{x}^*$. 

To address this instability, we compute for each feature its empirical selection frequency across the randomized local fits
$$
\hat{\pi}_j = \frac{1}{B}\sum_{b=1}^{B}
\mathbb{I}\left\{\hat{\beta}^{\bm{w}_b}_j \neq 0\right\},
$$
where $B$ is the number of refits of \eqref{eq:randomized_objective} and $\bm{w}_b$ denotes the randomization used in the $b$-th fit.
In particular, a low selection frequency signals uncertainty about a feature's contribution to the prediction at $\bm{x}^*$.

The table in Figure~\ref{fig:single_model} shows one such model for a single test car.
Alongside each coefficient we report its selection frequency, the fraction of randomized local refits retaining that feature.
In this example, five of the seven selected features are retained in over 90\% of refits; the least stable are cylinders and horsepower, retained in 75\% and 70\% respectively.
Randomization flags these as uncertain.
Here, with only $8$ features, most selections are stable across refits. 
This is not the case in higher dimensions: in the gene expression example of Section~\ref{sec:bctcga} ($p = 17{,}322$), the median local fit selects $94$ genes, of which typically only $15$ are retained in over $90\%$ of refits. In simulation (Appendix~\ref{app:simulation}), we find that the selection frequencies distinguish the true support while yielding a smaller, more interpretable set of features.

\begin{figure}[!t]
\centering
\begin{minipage}[t]{0.55\linewidth}
  \vspace{0pt}
  \includegraphics[width=\linewidth]{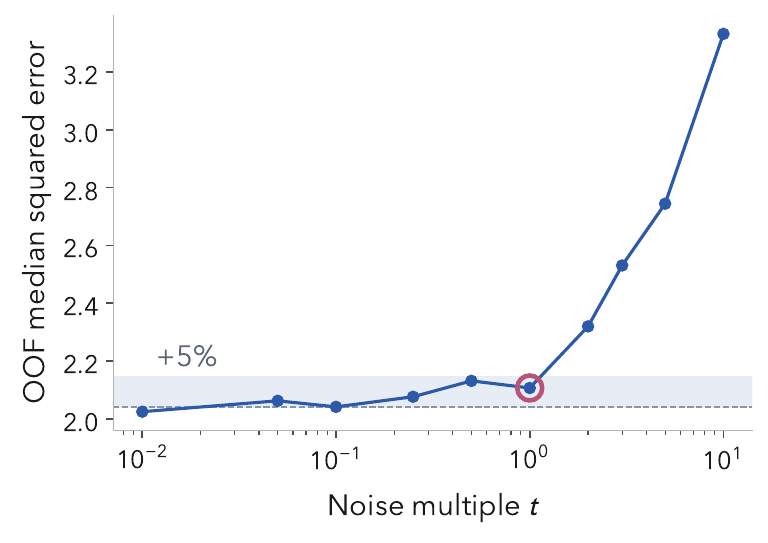}
\end{minipage}
\begin{minipage}[t]{0.4\linewidth}
  \centering
  \vspace{0.5em}
  {\small
  \sffamily
  \begin{tabular}{lrr}
  \toprule
  Feature & Coef. & Sel.\ prob. \\
  \midrule
  weight            & $-5.11$ & 100\% \\
  model year        & $2.28$  & 100\% \\
  origin: European  & $0.40$  & 99\% \\
  acceleration      & $0.52$  & 99\% \\
  origin: Japanese  & $0.41$  & 98\% \\
  \textbf{cylinders}  & $\bm{0.24}$  & \textbf{75\%} \\
  \textbf{horsepower} & $\bm{-0.17}$ & \textbf{70\%} \\
  \bottomrule
  \end{tabular}
  }
\end{minipage}
\caption{\textbf{Left: choosing the randomization scale $t$.} We take the largest $t$ whose out-of-fold median squared error remains within 5\% of the unperturbed fit (shaded). \textbf{Right: a local model for one test car} with mid-range fuel economy (distillation predicted $23.7$ mpg, observed $23.9$ mpg). The lasso selected 7 of the 8 features in the unperturbed fit. Selection probability is the fraction of 100 randomized refits retaining the feature at the selected scale $\hat t$ (circled at left). Features retained in under 90\% are shown in bold. Region coefficients are contrasts against American.}
\label{fig:single_model}
\end{figure}

Using randomized refits to address the instability of lasso-selected features is not new; a prominent example is stability selection \citep{meinshausen2010stability} based on subsampling.
The randomization mechanism in our approach, however, is deliberately different from subsampling.
This distinction has consequences for both the computational cost of local refitting and the interpretation of the resulting stability measures, which we discuss below.

{
\small
{
\begin{remark}
Across randomized refits of local distillation, as proposed in  \eqref{eq:randomized_objective}, the similarity weights and teacher predictions are computed once from the training data and then held fixed. 
Randomization perturbs only the optimization objective, keeping the design fixed and all $n$ observations in each refit. This can substantially reduce computation relative to subsampling-based methods such as stability selection, which require recomputing Steps 1 and 2 of Algorithm \ref{alg:local_distill} for every refit, as well as teacher predictions whenever they depend on the sampled training data, as with tabular foundation models using in-context learning.
\label{rem: stability sel1}
\end{remark}

\begin{remark}
The choice of randomization determines the notion of stability and thus the interpretation of the resulting predictions.
Beyond computational advantages, our mechanism preserves the locality learned from the training data: each refit keeps the similarity weights and teacher prediction anchor fixed. 
By contrast, subsampling can alter the local structure itself, especially when the effective sample size $\hat n_{\text{eff}}$ is small, i.e., when the weights $\hat S_j$ concentrate on few observations, or when noise is heterogeneous.
In such settings, variation across refits may reflect changes in the local structure rather than instability in the local feature--response relationship, making stability less reliable. 
\label{rem: stability sel2}
\end{remark}
}
}

\subsection{Interpreting the local models in aggregate}
\label{sec:interpretation_aggregate}

Beyond individual models, we are also interested in understanding heterogeneity within the sample. A natural approach is to fit locally distilled models for $n_{\text{test}}$ observations (or leave-one-out fits on the training data), and then cluster the fitted coefficient vectors into $k$ clusters. 
However, rather than identifying subgroups that may be artifacts of a single fit, we seek subgroups that are stable under small perturbations of the response $\bm{y}$. 
We therefore aggregate the clustering obtained across the randomized refits from Section~\ref{sec:randomized_estimator} using \emph{evidence accumulation clustering}~\citep{fred2005combining}, allowing the same refits used to characterize local feature--response relationships to also inform stable subgroup structure in the data.
We proceed as follows, illustrated on the Auto MPG data with $n_{\text{test}} = 157$ test cars:

\begin{enumerate}
  \item \textbf{Fit unperturbed local models:} run local distillation to obtain $n_{\text{test}} = 157$ fitted coefficient vectors in $\R^8$.
  \item \textbf{Choose the number of clusters:} do $k$-means clustering of the 157 vectors; the silhouette score selects $k=3$.
  \item \textbf{Cluster the randomized refits:} draw $100$ independent randomization vectors (Section~\ref{sec:randomized_estimator}). For each, solve \eqref{eq:randomized_objective} to get $157$ coefficient vectors and cluster them into $k = 3$ groups. This yields $100 \times 157$ cluster assignments.
  \item \textbf{Cluster the co-occurrence matrix:} the similarity between two test observations is the fraction of the $100$ clusterings placing them in the same group; use this to cluster once more into $k = 3$ groups using average linkage.
\end{enumerate}

Average linkage on the co-occurrence matrix yields the three groups in Figure~\ref{fig:coef_heatmap}, shown as a heatmap of the unperturbed coefficients ordered by cluster. 
The three clusters use different features: cylinders is predictive among the least efficient cars, and displacement among the most efficient.

\begin{figure}[!t]
  \centering
  \includegraphics[width=\linewidth]{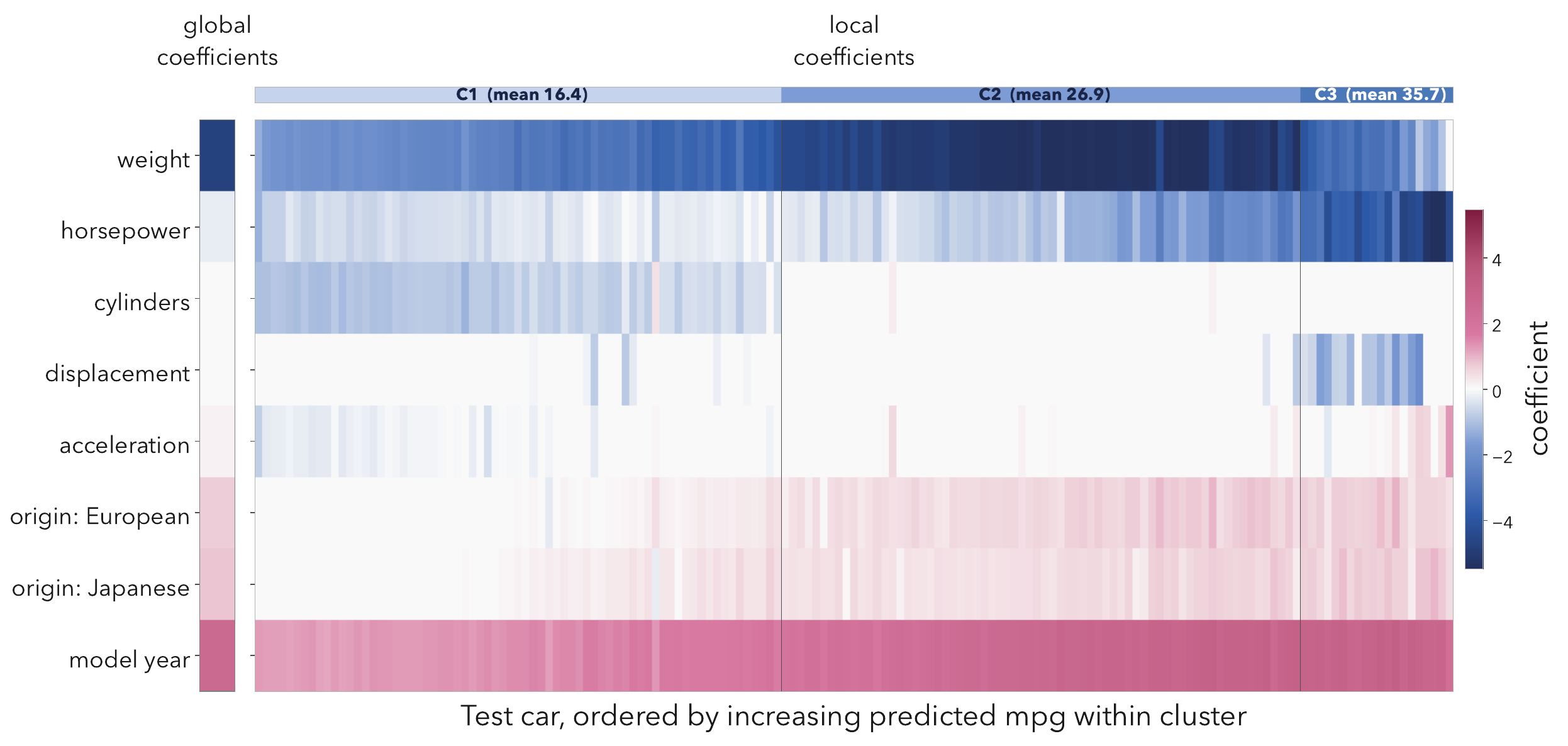}
  \caption{\textbf{Vehicle subgroups and their local miles per gallon models.} Local coefficients for the Auto MPG data, with test cars grouped into three subgroups (C1--C3) by the stable clustering of Section~\ref{sec:interpretation_aggregate}. Subgroups are ordered by mean predicted MPG, and cars by predicted MPG within subgroup.}
  \label{fig:coef_heatmap}
\end{figure}

\section{Theoretical analysis of stability} 
\label{sec:theory}

In this section, we establish stability guarantees for feature--selection probabilities under randomized local distillation \eqref{eq:randomized_objective}.
Throughout this section, we treat the similarity weights $\hat S_j$ and the regularization parameters $\hat\mu$ and $\hat\lambda$ as fixed. 
To avoid confusion, we drop the hats and write these quantities as $S_j$, $\mu$, and $\lambda$, respectively; similarly we write $n_{\mathrm{eff}}$ for the effective sample size.
We slightly modify our notation for the teacher's prediction at the query point, denoting it by $\hat\phi(\bm{x}^*; \bm{y})$, to make explicit that it may depend on the training response, as is the case in tabular foundation models, for example.
We focus on the lasso penalty, i.e., $\alpha=1$, while deferring an extension to the elastic-net penalty for future work to keep the theoretical development streamlined.
Our main result, Theorem \ref{thm:stability}, characterizes how the feature--selection probabilities are sensitive to changes in the input responses, leading to the uniform stability guarantee in Corollary \ref{cor:stability} under small perturbations of the training response.

\subsection{Notation and preliminaries}
First, let 
$ \bm{\ybar}
  = \bigl(y_1 \;\ldots y_n \; \hat\phi(\bm{x}^*;\bm{y})\bigr)^\top
  \in \R^{n+1}$,
and 
$\bm{\Xbar}
  = \bigl(\bm{x}_1\; \ldots\; \bm{x}_n \; \bm{x}^*\bigr)^\top
  \in \R^{(n+1)\times p}$
denote the augmented response vector and design matrix, obtained by appending the teacher's prediction $\hat\phi(\bm{x}^*; \bm{y})$ and the query point $\bm{x}^*$ to the response vector and design matrix, respectively.

Let $\bm{\Om} \;=\; \diag\!\left(S_1, \dots, S_n, S_{n+1}\right)\in \R^{(n+1)\times(n+1)}$ denote the diagonal matrix of weights for the $n+1$ observations in the augmented dataset, where $S_{n+1}=\frac{\mu}{\sqrt{n_{\mathrm{eff}}}}$ is the weight assigned to the pseudo-response at the query point.
We then define
\begin{equation}
\label{eq:whitened}
  \bm{z}= \bm{z}(y_1,\ldots, y_n) \;=\; \bm{\Om}^{1/2}\,\bm{\ybar} \in \R^{n+1},
  \qquad
  \bm{V} \;=\; \bm{\Om}^{1/2}\,\bm{\Xbar} \in \R^{(n+1)\times p} ,
\end{equation}
to be the weight-adjusted response and design, respectively, i.e.,  coordinatewise, $Z_k = \sqrt{S_k}\, y_k \quad (k \le n)$, and $Z_{n+1} = \sqrt{S_{n+1}} \, \hat\phi(\bm{x}^*; \bm{y}) = \frac{\sqrt{\mu}}{\neff^{1/4}}\, \hat\phi(\bm{x}^*; \bm{y})$.

Under this notation, when $\alpha=1$ and $\lambda$ denotes the tuning parameter for the lasso penalty, the randomized local distillation problem \eqref{eq:randomized_objective} is equivalent to solving 
\begin{equation}
\label{eq:randomized-lasso}
  \bm{\bhat}^{\,\bm{\omega}}
  \;=\;
  \operatorname*{argmin}_{\bm{\beta} \in \R^{p}}
  \; \tfrac{1}{2}\,\norm{\bm{z} + \bm{\omega} - \bm{V}\bm{\beta}}_{2}^{2}
  \;+\; \lambda \norm{\bm{\beta}}_{1},
\end{equation}
where $\bm{\omega}\sim \mathcal{N}(0, \tau^2 \bm{I}_{n+1})$ is drawn independently of the training response.
We adopt this equivalent formulation throughout the section, as it allows us to first establish guarantees for randomized lasso regression with a perturbed response, which may be of independent interest beyond this setting.

Let $\widehat E_V(\boldsymbol z,\boldsymbol\omega)=\operatorname{supp}\!\left(\widehat{\boldsymbol\beta}^{\boldsymbol\omega}(\boldsymbol z)\right)$
denote the active set obtained by solving the randomized lasso in \eqref{eq:randomized-lasso}.
The quantity of interest is the feature-selection probability for predictor $j$, evaluated as a function of $\bm{y}$ and denoted by
\begin{equation}
\label{eq:pi-def}
  \pi_j(\bm{y})
  \;=\;
  \PPo\!\left[\, j \in \Ehat_V(\bm{z},\bm{\omega})\right],
  \quad \text{ for } j \in [p],
\end{equation}
where the probability is over the randomization $\bm{\omega}$, with $\bm{y}$ and hence $\bm{z}=\bm{z}(y_1,\ldots, y_n)$ treated as fixed, and $[p]=\{1,2,\ldots, p\}$. 
The subscript $\bm{\omega}$ on the probability, here and throughout, indicates that the probability is taken only with respect to the randomization.

\subsection{Stability of randomized local distillation}
\label{sec:stability-main}

Because the selection probabilities in \eqref{eq:pi-def} are obtained by convolving the discontinuous feature--selection indicator with the Gaussian density of $\bm{w}$, they are smooth provided that the teacher's prediction is a smooth function of the input $\bm{y}$. 
Lemma \ref{lem:smoothness} formalizes this observation; the proof is provided in Appendix~\ref{App: main}.

\begin{lemma}[Smoothness of feature--selection probabilities]
Fix $k\in\mathbb{N}\cup\{\infty\}$ and suppose that $\bm y\mapsto\hat\phi(\bm{x}^*; \bm{y})$ belongs to $\mathcal{C}^k(\R^n)$. Then for every
$j\in[p]$, the selection probability $\pi_j(\bm{y})$, defined in \eqref{eq:pi-def}, also belongs to $\mathcal{C}^k(\R^n)$.
\label{lem:smoothness}
\end{lemma}

We now state Theorem \ref{thm:stability}, which bounds the sensitivity of the feature--selection probabilities to perturbations of the training response $\bm{y}$, leveraging the smoothness of these probabilities established in Lemma \ref{lem:smoothness}.
We derive this result under the following assumptions on the design matrix $\bm{X}$, the similarity weights $S_j$ and the weight-adjusted design matrix $\bm{V}$.

\begin{assumption}[Weight-adjusted design in general position]
\label{assump:general-position}
The columns of $\bm{V}$ are in \textit{general position} \citep{Tibshirani2012TheLP}.
\end{assumption}

\begin{assumption}[Bounded weight concentration]
\label{assump:distillation-weight-spikiness}
We assume that there exists a constant $C_S<\infty$ such that $n_{\mathrm{eff}}S_{\max}\le C_S$.
\end{assumption}

\begin{assumption}[Bounded design and local distillation weight]
\label{assump:distillation-bounded-covariates}
For some constants $B_X,C_\mu>0$, we assume that
$\max\left\{\max_{i\in[n]}\|\boldsymbol x_i\|_\infty, \|\bm{x}^*\|_\infty \right\}\le B_X$, and $0<\mu\le C_\mu$.
\end{assumption}

\begin{assumption}[Non-degeneracy of restricted design]
\label{assump:distillation-local-nondegeneracy}
For some constant $\kappa_X>0$, we assume that, uniformly over $j\in[p]$ and $E\in\mathcal E_{-j}^V$, where $\mathcal E_{-j}^V$ denotes the collection of essential active sets from lasso regression on $\bm{V}_{-j}$ (i.e., active sets from the leave-$j$-covariate-out lasso fit whose corresponding selection regions have positive Lebesgue measure),
\[\sigma_{\min}\left(\frac{1}{\sqrt{|\mathcal I_S|}}\,\boldsymbol X_{\mathcal I_S,E\cup\{j\}}\right)\ge \kappa_X,\]
where $\mathcal I_S=\left\{i\in[n]:S_i\ge\frac{1}{2n_{\mathrm{eff}}}\right\}$ collects the training observations whose similarity weights are at least one half of the effective uniform weight.
\end{assumption}

Under Assumption~\ref{assump:general-position}, the randomized local distillation problem \eqref{eq:randomized-lasso} admits a unique solution.
Assumption~\ref{assump:distillation-weight-spikiness} prevents the similarity weights from concentrating too heavily on a small number of observations. Assumption~\ref{assump:distillation-bounded-covariates} requires both the design covariates and the distillation-weight parameter to be bounded. Assumption~\ref{assump:distillation-local-nondegeneracy} imposes a lower singular-value condition on the design restricted to the effective similarity neighborhood $\mathcal I_S$, preventing the relevant design columns from becoming nearly linearly dependent among observations receiving non-negligible weights.

\begin{theorem}[Sensitivity bound for selection probabilities]
\label{thm:stability}
Suppose that $\bm y\mapsto\hat\phi(\bm{x}^*; \bm{y})$ belongs to $\mathcal{C}^1(\R^n)$ with $L_{\hat\phi}(\bm{x}^*) = \sup_{\bm{y}\in \R^n} \|\nabla \hat\phi(\bm{x}^*; \bm{y})\|_{\infty}<\infty$, and that Assumptions \ref{assump:general-position}, \ref{assump:distillation-weight-spikiness}, \ref{assump:distillation-bounded-covariates} and \ref{assump:distillation-local-nondegeneracy} hold.
Then there exists a constant $C<\infty$ such that, for every $j\in[p]$, 
\begin{equation*}
  \sup_{\bm{y}\in \R^n} \|\nabla \pi_j(\bm{y})\|_{\infty}
  \;\le\;
  \sqrt{\frac{2}{\pi}}\frac{C}{\tau}
  \left(\frac{1}{\sqrt{\neff}}
    \;+\; \frac{L_{\hat\phi}(\bm{x}^*)}{\neff^{1/4}}\right).
\end{equation*}
\end{theorem}

A proof of Theorem \ref{thm:stability} is provided in Appendix~\ref{App: main}.
The proof consists of two main components. 
First, we analyze the stability of the randomized lasso obtained by regressing a perturbed response on a $p$-dimensional feature matrix, as detailed in Appendix~\ref{App:stability:lasso}. 
Second, we apply the chain rule to characterize the additional contribution from local distillation, as detailed in Appendix~\ref{App:stability:aux}.

We make a few observations on the above-stated result.
\begin{enumerate}
\item[(i)] It follows directly from the bound in Theorem \ref{thm:stability} that the contribution of each training observation to the sensitivity of a feature's selection probability has two components: the first is a direct contribution, excluding the effect of distillation, which is of order $\neff^{-1/2}$, and the second contribution arises through the teacher's prediction, which is of order $\neff^{-1/4} L_{\hat\phi}(\bm{x}^*)$.
We make no assumption on $L_{\hat\phi}(\bm{x}^*)$; if it decreases with $\neff$, the second term achieves a sharper rate.
\item[(ii)] When the teacher is a fully pretrained model (i.e., the teacher's prediction does not depend on the training samples), the bound in Theorem \ref{thm:stability} simplifies to 
$$\sup_{\bm{y}\in \R^n} \|\nabla \pi_j(\bm{y})\|_{\infty}\;\le\;\sqrt{\frac{2}{\pi}}\frac{C}{\tau} \neff^{-1/2},$$
since $L_{\hat\phi}(\bm{x}^*)=0$. 
As a further special case, for the global lasso fit without using the teacher as a Bayesian prior, and with uniform similarity weights $S_j= 1/n$, the same bound holds with $\neff=n$.
\end{enumerate}

Theorem \ref{thm:stability} immediately implies the following corollary on the stability of feature--selection probabilities under a single-response perturbation, analogous to algorithmic stability \citep{bousquet2002stability}; the proof is omitted as it follows directly from the sensitivity bound in the theorem.

\begin{corollary}[Stability guarantee]
Under the assumptions of Theorem \ref{thm:stability}, there exists a constant $C<\infty$ such that, for every $j\in[p]$, $i\in[n]$, and for any $\bm{y},\bm{y}'\in\mathbb{R}^n$ satisfying $y_k=y_k'$ for all $k\neq i$ (equivalently, differing only in the $i$-th coordinate),
$$\bigl\lvert \pi_j(\bm{y}) - \pi_j(\bm{y}') \bigr\rvert\;\le \sqrt{\frac{2}{\pi}}\frac{C}{\tau}
  \left(\frac{1}{\sqrt{\neff}}
    \;+\; \frac{L_{\hat\phi}(\bm{x}^*)}{\neff^{1/4}}\right) |y_i - y'_i|.$$ 
\label{cor:stability}
\end{corollary}


The stability bound in Corollary \ref{cor:stability} improves with $\tau$, but excessive randomization can reduce predictive accuracy. 
Thus, as described in Section \ref{sec:randomized_estimator}, we choose $\tau$ to balance stability with a user-specified tolerance for predictive accuracy loss.

\noindent\textbf{Extension to stability over a regularization grid.} Our stability-based approach could be extended to base decisions on a grid $\Lambda$ of lasso regularization parameters $\lambda$, rather than on a single fixed value of $\lambda$.
As proposed in the stability selection framework of \citet{meinshausen2010stability}, a natural quantity on which to base decisions is $\displaystyle\max_{\lambda\in\Lambda}\pi_j(\bm{y}; \lambda)$, where $\pi_j(\bm{y}; \lambda)$ denotes the selection probability at $\lambda$ (with the dependence on $\lambda$ made explicit in the notation), and the maximum taken over the grid of regularization parameters.
However, our proof techniques for controlling the sensitivity of this quantity to perturbations of the training input, and hence for establishing stability guarantees, may not extend directly to the maximum operation because it is not smooth.

Instead, one can consider a smooth approximation to $\displaystyle\max_{\lambda\in\Lambda}\pi_j(\bm{y}; \lambda)$, namely $\Pi_{\Lambda}(\bm{y})=\frac{1}{\eta}\log\left(\frac{1}{|\Lambda|}\sum_{\lambda\in\Lambda}\exp\left(\eta\pi_j(\bm{y}; \lambda)\right)\right)$, for fixed $\eta \in \mathbb{R}^{+}$. 
Like the selection probability for any fixed $\lambda$, this quantity also takes values in $[0,1]$.
In practice, using its empirical counterpart $\widehat\Pi_{\Lambda}(\bm{y})$, one can then determine the set of selected features as $\{j: \widehat{\Pi}_{\Lambda}(\bm{y}) \geq p_{\text{thr}}\}$, for a chosen threshold $p_{\text{thr}}$.
 Furthermore, the gradient of $\Pi_{\Lambda}(\bm{y})$ can be written as
$$\nabla\Pi_{\Lambda}(\bm y)
=
\sum_{\lambda\in\Lambda}
w_{\lambda}(\bm y)\,
\nabla\pi_j(\bm y;\lambda), \quad \text{with } \ w_{\lambda}(\bm y)
=
\dfrac{\exp\!\left(\eta\pi_j(\bm y;\lambda)\right)}
{\sum_{\lambda'\in\Lambda}\exp\!\left(\eta\pi_j(\bm y;\lambda')\right)},
$$
a convex combination of the derivatives of the selection probabilities. 
Consequently, a fairly straightforward extension of our results to establish a uniform stability bound over the grid $\Lambda$ would yield a corresponding stability guarantee for $\Pi_{\Lambda}(\bm{y})$ under perturbations to a single coordinate of the input response, analogous to the bound in Corollary \ref{cor:stability}.
We omit the details of this extension from the present work.

\section{A high-dimensional example: predicting gene expression}
\label{sec:bctcga}

Here, we turn to a breast cancer gene expression study from The Cancer Genome Atlas (TCGA)~\citep{tcga2012}. Tumor samples from $n = 536$ patients were assayed on Agilent mRNA expression microarrays, and measurements are on the log scale. Following the example in \citet{breheny2025hdrm}, we treat BRCA1 expression as the response and the remaining $p = 17{,}322$ genes as predictors, excluding 491 genes with missing data. BRCA1 is the first gene identified whose mutations increase the risk of early onset breast cancer. Because BRCA1 is likely to interact with many other genes, genes whose expression correlates with it are candidates for further study.

We divide the data into a 60/40 train/test split, with $n_{\text{train}} = 321$ and $n_{\text{test}} = 215$. The student model is lasso regression. 
For the teacher, we consider two candidates and choose between them using cross-validated error: TabPFN applied to all $17{,}322$ genes, and TabPFN applied to the $500$ genes most correlated with BRCA1 expression, screened within each training fold. 
The screened teacher performed far better and was selected. (The global lasso did not improve when restricted to the $500$ gene subset.) 
Note that in this setting, the teacher and student have different feature representations, as in cross-modal distillation: the teacher uses $500$ genes while the student uses all $17{,}322$. 
With the training set, we estimated $\hat \mu = 1.39$, and the median effective sample size was 211 out of 321 (minimum 3, maximum 237, mean 192).

On this split, the global lasso had test PSE $0.189$ ($R^2 = 0.579$), TabPFN $0.150$ ($R^2 = 0.666$), and local distillation $0.148$ ($R^2 = 0.670$): a 22\% reduction relative to the global lasso, matching (here, slightly exceeding) its teacher, while retaining transparency. Additionally, the local fits were sparser than the global lasso fit: the global fit selected $123$ genes, while the median local fit selected $94$.%
\footnote{The PSE improvement was robust across 100 random 60/40 train/test splits, where local distillation improved on the global lasso in 96\% of runs, with a median PSE reduction of 22\%. The sparsity comparison was also stable: the median local fit was sparser than the global fit on 97 of 100 splits, with a median reduction of 15\% in the number of selected genes.}  
Applying the stability screen of Section~\ref{sec:interpretation} ($\hat\pi_j > 0.9$ across $100$ randomized refits) reduced this further: the local models retained a median of $15$ stably selected genes (IQR $13$–$17$).

 We then cluster the local models using 100 randomized fits with randomization scale $t = 0.1$ and $k=5$ clusters (parameters selected as in Section~\ref{sec:interpretation_aggregate}). Figure~\ref{fig:brca1_clusters} reveals heterogeneity across patients: FAM107A and KLF14 are selected primarily in clusters one and five respectively, both with negative coefficients, and both given a coefficient of zero by the global lasso. These are reported tumor suppressors down-regulated in cancer \citep{ou2022identification, chu2022klf14}. This heterogeneity is not visible to a global linear model, and is difficult to surface in black-box models.

\begin{figure}[!t]
  \centering
  \includegraphics[width=0.9\linewidth]{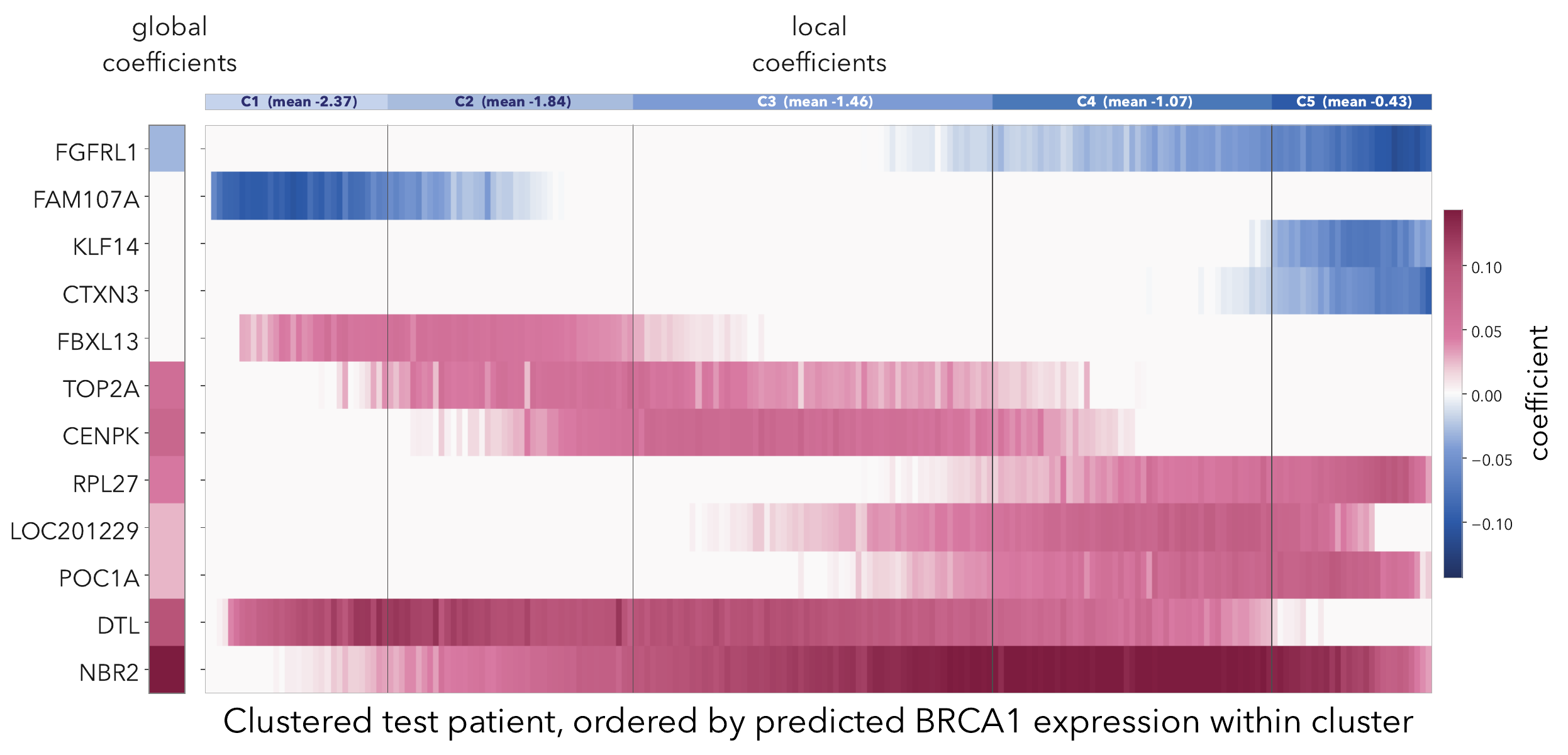}
  \caption{\textbf{Local BRCA1-model coefficients across patient clusters.} Columns are the 215 test patients; rows are the genes whose local coefficients vary most across clusters, and two of the largest-magnitude genes for reference.  Patient clustering used $100$ randomized fits (Section~\ref{sec:interpretation_aggregate}). The leftmost column is the global lasso fit on the same training data, for comparison.}
  \label{fig:brca1_clusters}
\end{figure}

\section{Benchmark comparisons}
\label{sec:real_data}

Here, we use common machine learning benchmark datasets to compare local distillation to (1) the student model (global lasso or ridge regression), (2) the teacher model (TabPFN, XGBoost~\citep{xgboost}, or TabFM~\citep{kong2026tabfm}) and (3) two local linear models (LOESS and local linear forests).  We evaluate on 17 regression datasets spanning sample sizes $n \in [159, 4177]$ and feature counts $p \in [5, 51]$. The datasets are from the \textbf{UCI Machine Learning Repository} \citep{UCIRepository} ({\em automobile, servo, liver disorders, auto MPG, real estate valuation, infrared thermography temperature, student performance}) and the \textbf{OpenML-CTR23 regression benchmark} \citep{bischl2017openml, fischer2023openml} ({\em cars, QSAR fish toxicity, concrete compressive strength, socmob, airfoil self-noise, red wine, auction verification, space ga, abalone, white wine}).

For each dataset, we generate 20 random 80/20 train/test splits. We use a complete-case design and perform one-hot encoding for categorical variables. Within each train/test split, we normalize features using the mean and standard deviation of the training set. Per-dataset sample size and feature counts are in Table~\ref{tab:datasets}.

With TabPFN as teacher and a lasso student, local distillation improved on the global lasso on all 17 datasets (median gain 0.11 in test $R^2$) and recovered a median 94\% of the teacher's improvement over the lasso (range 72--104\%); the median gap to the teacher was 0.006 in $R^2$, and under 0.02 on 14/17 datasets. Figure~\ref{fig:r2_closeup} shows six representative datasets; results for XGBoost and TabFM teachers, ridge students, and the remaining datasets are in Appendix~\ref{app:real_data}.


\begin{figure}[!t]
    \centering 
    \includegraphics[width=\linewidth]{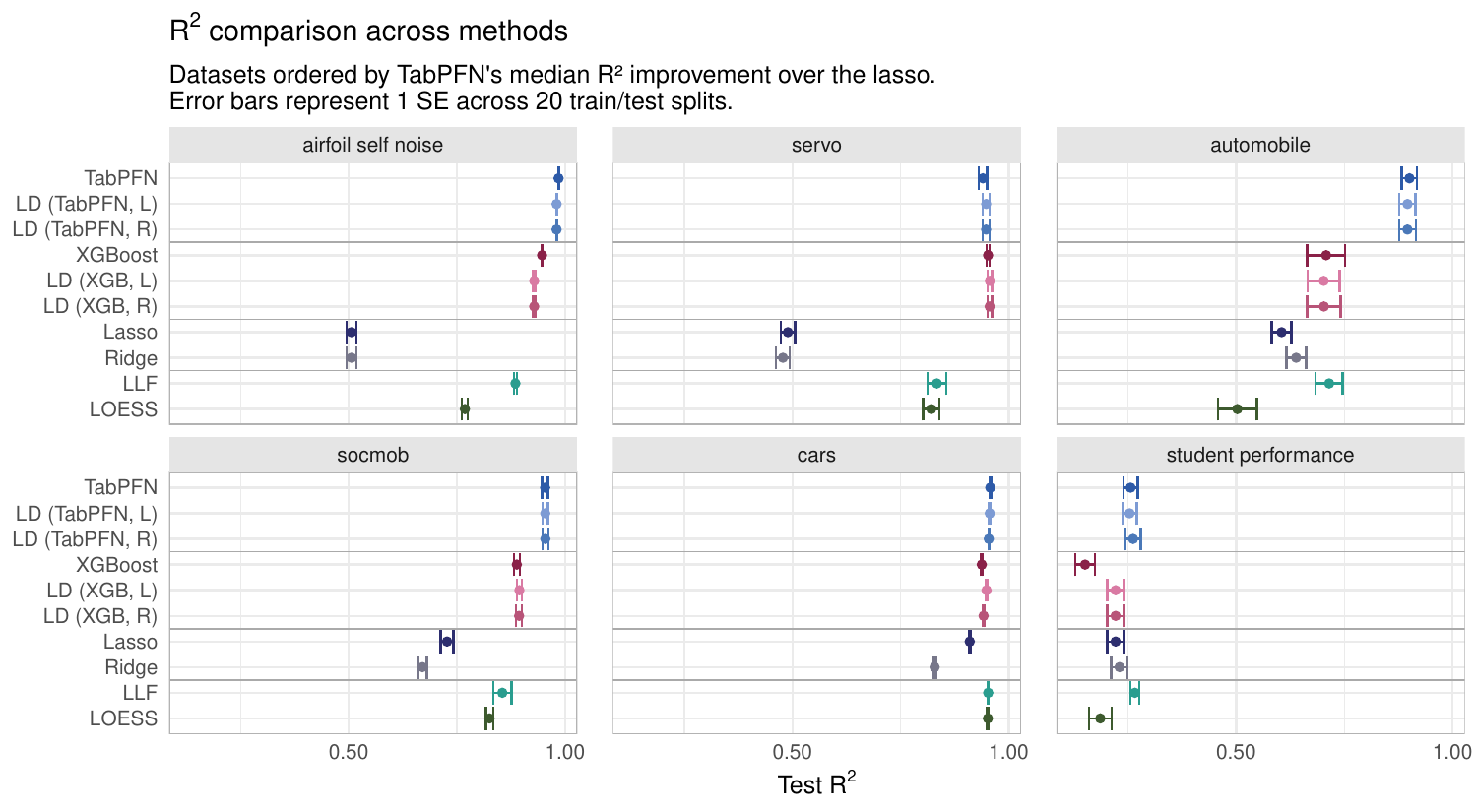}
      \caption{\textbf{Predictive performance (test $R^2$, median $\pm$ 1 standard error) for six datasets from the UCI ML and OpenML repositories.} Local distillation is labeled as ``LD (teacher, regularization)'', using ``L'' for lasso and ``R'' for ridge. When the teacher outperforms the global linear model, local distillation usually approaches the teacher's performance. More datasets are included in Appendix~\ref{app:real_data}.}
    \label{fig:r2_closeup}
\end{figure}  

\section{Discussion}
\label{sec:discussion}

Local linear models are in a ``sweet spot'' between interpretable modeling and modern AI: they retain the benefits of linear models (transparent, easy to understand, computationally simple) while rivaling the performance of black-box models.
The main challenges are defining what constitutes \emph{local}, fitting the local models, and interpreting the resulting collection of fits. Local distillation addresses the first two using a black-box teacher and the third through randomized refitting. Across 17 benchmark datasets, a high-dimensional gene expression example, and three teacher models (TabPFN, TabFM, and XGBoost), local distillation consistently matches or approaches the predictive accuracy of its teacher.

When a linear student cannot match the teacher's accuracy, a richer student class incorporating interactions or covariate transformations may narrow the gap. With multiple candidate teachers, the teacher can be selected by cross-validation on the training data and, because it enters only through its predictions, may use a different feature set or modality from the student. Our gene expression example illustrates this flexibility and suggests cross-modal distillation as a promising direction.

The distillation strength $\hat\mu$ is a ratio of cross-validated losses; this choice is intuitive and performed well empirically, though we make no claim of optimality. We also revert to the global linear model when $\hat \mu \leq 1$. A smoother alternative based on the teacher's \textit{excess} performance, $(\hat\mu - 1)_+$, reduced accuracy, suggesting that when the teacher outperforms the student, the local fits benefit from assigning it substantial weight, consistent with \citet{hinton2015distilling}. Alternative estimators of $\hat\mu$ remain worth exploring.

We view local linear modeling as a flexible, modular framework that can compete with modern predictive methods: the similarity weights define locality, the penalty defines structure, and pseudo-observations carry external information into the fit.
The choice of weights deserves particular care, because the definition of ``local'' determines \textit{which} heterogeneity the local models can express and the appropriate notion of stability; as in clustering, many choices are valid though not all are informative.  
Local distillation is one concrete realization of this framework, using a black-box model to guide a fairly simple model that rivals its teacher in predictive accuracy while yielding trustworthy interpretations through its coefficients.

{\bf Acknowledgements}. We would like to thank Robert Tibshirani, Trevor Hastie, Liza Levina, and Shihan Khan for helpful comments. The authors used Claude Opus 4.5 (model ID: claude-opus-4-5-20251101) and ChatGPT 5.0 for coding support and text editing. The gene expression example shown here is based upon data generated by the TCGA Research Network: \url{https://www.cancer.gov/tcga}.

{\bf Funding}. S.P. was supported by NSF CAREER Award DMS-2337882.

{\bf Disclosure}. The authors report there are no competing interests to declare.

{\bf Data availability}. The data that support this work are public and cited throughout. The scripts to download data and run simulations are published on GitHub at \url{https://github.com/erincr/local-distillation-benchmark}. The breast cancer expression data are from The Cancer Genome Atlas (\url{https://www.cancer.gov/tcga}); we use the processed version distributed with the \texttt{hdrm} R package (\url{https://github.com/pbreheny/hdrm}).

\appendix

\section{Details of theoretical results}

\subsection{Proof of main results}
\label{App: main}

\begin{proof}[Proof of Lemma \ref{lem:smoothness}]
Let $\pi_j^{\circ}(\bm{z})=\PPo[\,j\in\Ehat_V(\bm{z},\bm{\omega})]$ denote the selection probability of predictor $j$ under lasso regression of $\bm{z}+\bm{\omega}$ on $\bm{V}$, evaluated as a function of the weight-adjusted response $\bm{z}$.
Then, by definition, $$\pi_j(\bm{y})=\pi_j^{\circ}\bigl(\bm{z}(\bm{y})\bigr) = \pi_j^{\circ} \circ \bm{z}(\bm{y}).$$

The map $\bm{y} \mapsto \bm{z}(\bm{y})$ has entries $\sqrt{S_k}\,y_k$ for
$k\le n$ and $\sqrt{\mu}\,n_{\mathrm{eff}}^{-1/4}\hat\phi(\bm x^{*};\bm y)$ for $k=n+1$; the first $n$ entries are linear in $\bm y$, while the final entry is $k$ times differentiable by assumption, and hence, $\bm z:\R^{n}\to\R^{n+1}$ is $k$ times differentiable. 
 
Observe that we can write
\[
  \pi_j^{\circ}(\bm z)
  :=\PPo\bigl[\,j\in\Ehat_V(\bm z,\bm\omega)\bigr]
  =\int_{\R^{n+1}}\chi_j(\bm z+\bm w)\,\varphi_\tau(\bm w)\,d\bm w
  =(\chi_j*\varphi_\tau)(\bm z),
\]
where $\chi_j(\bm t)=\mathbf 1\{\,j\in\Ehat(\bm t)\,\}$, and
$\varphi_\tau(\bm u)$ is the $N(\bm 0,\tau^{2}\bm I_{n+1})$ density function at $\bm u$. 

Since $\chi_j$ is bounded and $\varphi_\tau$ has integrable derivatives of all orders, differentiation under the integral sign yields $\pi_j^{\circ}\in C^{\infty}(\R^{n+1})$, with $\partial^{\alpha}\pi_j^{\circ}=\chi_j*\partial^{\alpha}\varphi_\tau$ for every multi-index $\alpha$. Applying the chain rule to $\pi_j=\pi_j^{\circ}\circ\bm z$ shows that, for each $r\le k$, every derivative of $\pi_j$ of order $r$ can be expressed in terms of derivatives of $\pi_j^{\circ}$ and derivatives of $\bm z$ of order at most $r$. Since all such derivatives exist, $\pi_j$ is $k$ times differentiable on $\R^{n}$.
\end{proof}

\begin{proof}[Proof of Theorem \ref{thm:stability}]
Combining Proposition~\ref{thm:distillation-coordinate-stability} with Lemma~\ref{lem:distillation-native-comparison-bounds} (both stated in Section \ref{App:stability:aux}) yields, for any $y\in \R^n$ and $i\in [n]$, that
\begin{align*}
|\partial_i \pi_j(\bm{y})|
&\le \sqrt{\frac{2}{\pi}}\,\frac{1}{\tau}
\left\{ \sqrt{S_i}\,\rho^{{V}}_{ij}
+ \sqrt{S_{n+1}}\,|\partial_i \hat\phi(\bm{x}^*; \bm{y})|\,\rho^{{V}}_{n+1,j} \right\} \\
&\le \sqrt{\frac{2}{\pi}}\,\frac{C}{\tau}
\left\{ \frac{1}{\sqrt{n_{\mathrm{eff}}}}
+ \frac{|\partial_i \hat\phi(\bm{x}^*; \bm{y})|}{n_{\mathrm{eff}}^{1/4}} \right\}.
\end{align*}
Taking supremum over $y$ and $i\in [n]$ proves our claim.
\end{proof}

\subsection{Theoretical analysis of stability under randomized lasso}
\label{App:stability:lasso}

Throughout this section, we study the stability of the feature--selection probabilities 
\[
    \pi_j^{\circ}(\boldsymbol z)
    =
    \mathbb P_{\boldsymbol\omega}
   \left[j\in\widehat E_V(\boldsymbol z,\boldsymbol\omega)\right].
\]
with respect to perturbations in $\bm{z}$.
The main result of this section, Theorem \ref{thm:single-feature-stability}, characterizes the sensitivity of feature-selection probabilities to perturbations in the response under randomized lasso regression.

\subsubsection{Leave-$j$-out score and lasso geometry}

For a feature $j\in[p]$, define the leave-$j$-out lasso
\begin{equation}
    \widehat{\boldsymbol\beta}_V^{(-j)}(\boldsymbol z)
    =
    \argmin_{\boldsymbol b\in\mathbb R^{p-1}}
    \left\{
        \frac12\|\boldsymbol z-\boldsymbol V_{-j}\boldsymbol b\|_2^2+\lambda\|\boldsymbol b\|_1
    \right\}.
    \label{eq:leave-j-out-lasso}
\end{equation}
Furthermore, let $\boldsymbol R_{-j}^V(\boldsymbol z)=\boldsymbol z-\boldsymbol V_{-j}\widehat{\boldsymbol\beta}_V^{(-j)}(\boldsymbol z)$
be the leave-$j$-out residual, and define
\begin{equation}
    T_j^V(\boldsymbol z)=\boldsymbol V_j^\top \boldsymbol R_{-j}^V(\boldsymbol z).
    \label{defn:Tj}
\end{equation}

\begin{lemma}[KKT characterization]
\label{lem:leave-j-out-kkt}
Under Assumption \ref{assump:general-position}, it holds that
\[
    j\notin\widehat E_V(\boldsymbol z,\boldsymbol\omega)
    \quad\Longleftrightarrow\quad
    |T_j^V(\boldsymbol z+\boldsymbol\omega)|\le\lambda,
\]
where $T_j^V(\cdot)$ is as defined in \eqref{defn:Tj}.
\end{lemma}

\begin{proof}
Set $\boldsymbol t=\boldsymbol z+\boldsymbol\omega$. By definition,
$\widehat E_V(\boldsymbol z,\boldsymbol\omega)$ is the support of the
full Lasso solution with response $\boldsymbol t$.
The KKT conditions for the full Lasso problem state that a vector
$\boldsymbol{\beta}\in\mathbb R^p$ is optimal if and only if there exists
$\boldsymbol{\gamma}\in\mathbb R^p$ such that
\begin{equation}
    \boldsymbol V^\top(\boldsymbol t-\boldsymbol V\boldsymbol{\beta})=\lambda\boldsymbol{\gamma},
    \label{eq:full-lasso-kkt}
\end{equation}
where, for every $k\in[p]$,
$ \boldsymbol{\gamma}_k=
    \begin{cases}
        \operatorname{sign}(\boldsymbol{\beta}_k), & \boldsymbol{\beta}_k\ne0,\\
        u_k\in[-1,1], & \boldsymbol{\beta}_k=0.
    \end{cases}
$

Let $\widehat{\boldsymbol\beta}_V^{(-j)}(\boldsymbol t)$ be the solution of the leave-$j$-out
Lasso, and define
$\boldsymbol R_{-j}^V(\boldsymbol t)
    =
    \boldsymbol t-\boldsymbol V_{-j}\widehat{\boldsymbol\beta}_V^{(-j)}(\boldsymbol t).
$
The KKT conditions for the restricted problem imply that there exists
$\boldsymbol{\gamma}_{-j}\in\mathbb R^{p-1}$ such that
\begin{equation}
    \boldsymbol V_{-j}^\top \boldsymbol R_{-j}^V(\boldsymbol t)
    =
    \lambda\boldsymbol{\gamma}_{-j},
    \label{eq:restricted-lasso-kkt}
\end{equation}
with
\[
    (\boldsymbol{\gamma}_{-j})_k
    =
    \begin{cases}
        \operatorname{sign}
        \bigl(\widehat{\boldsymbol\beta}_{V,k}^{(-j)}(\boldsymbol t)\bigr),
        &
        \widehat{\boldsymbol\beta}_{V,k}^{(-j)}(\boldsymbol t)\ne0,\\
        u_k\in[-1,1],
        &
        \widehat{\boldsymbol\beta}_{V,k}^{(-j)}(\boldsymbol t)=0.
    \end{cases}
\]

Now construct the $p$-dimensional candidate $\widetilde{\boldsymbol\beta}$ with
\[
    \widetilde{\boldsymbol\beta}_j=0,
    \qquad
    \widetilde{\boldsymbol\beta}_{-j}
    =
    \widehat{\boldsymbol\beta}_V^{(-j)}(\boldsymbol t).
\]
Its residual in the full Lasso problem is exactly
$
    \boldsymbol t-\boldsymbol V\widetilde{\boldsymbol\beta}
    =
    \boldsymbol R_{-j}^V(\boldsymbol t).
$
Equation~\eqref{eq:restricted-lasso-kkt} verifies the full KKT
conditions for every coordinate other than $j$. 
Since
$\widetilde{\boldsymbol\beta}_j=0$, the remaining KKT condition is
$
    \boldsymbol V_j^\top \boldsymbol R_{-j}^V(\boldsymbol t)\in\lambda[-1,1].$
By definition,
$
    T_j^V(\boldsymbol t)=\boldsymbol V_j^\top \boldsymbol R_{-j}^V(\boldsymbol t),
$
which gives
$
    |T_j^V(\boldsymbol t)|\le\lambda.
$
Therefore,
\begin{equation}
    |T_j^V(\boldsymbol t)|\le\lambda
    \quad\Longleftrightarrow\quad
    \widetilde{\boldsymbol\beta}
    \text{ is a solution of the full Lasso problem}.
    \label{eq:candidate-full-solution}
\end{equation}

If $|T_j^V(\boldsymbol t)|\le\lambda$, then
\eqref{eq:candidate-full-solution} shows that $\widetilde{\boldsymbol\beta}$ is a
full Lasso solution. By uniqueness of the full Lasso solution, guaranteed  by Assumption \ref{assump:general-position},
$
    \widehat{\boldsymbol\beta}(\boldsymbol t)=\widetilde{\boldsymbol\beta}.$
Since $\widetilde{\boldsymbol\beta}_j=0$, it follows that
$j\notin\widehat E_V(\boldsymbol z,\boldsymbol\omega)$.

Conversely, suppose that $j\notin\widehat E_V(\boldsymbol z,\boldsymbol\omega)$, so that
$\widehat{\boldsymbol\beta}_j(\boldsymbol t)=0$. The vector $\widehat{\boldsymbol\beta}_{-j}(\boldsymbol t)$ must then
solve the leave-$j$-out problem. Indeed, if some
$\boldsymbol b\in\mathbb R^{p-1}$ had a strictly smaller restricted objective,
then the full vector $(\boldsymbol b,0)$ would have a strictly smaller full Lasso
objective than $\widehat{\boldsymbol\beta}(\boldsymbol t)$, contradicting its optimality.
By uniqueness of the restricted solution, implied by Assumption \ref{assump:general-position},
$
    \widehat{\boldsymbol\beta}_{-j}(\boldsymbol t)
    =
    \widehat{\boldsymbol\beta}_V^{(-j)}(\boldsymbol t).
$
The full KKT condition at the zero coefficient
$\widehat{\boldsymbol\beta}_j(\boldsymbol t)=0$ therefore gives
\[
    \left|
        \boldsymbol V_j^\top
        \left(
            \boldsymbol t-\boldsymbol V_{-j}\widehat{\boldsymbol\beta}_V^{(-j)}(\boldsymbol t)
        \right)
    \right|
    \le\lambda.
\]
Equivalently, $|T_j^V(\boldsymbol t)|\le\lambda$. 
We have thus proved
\[
    j\notin\widehat E_V(\boldsymbol z,\boldsymbol\omega)
    \quad\Longleftrightarrow\quad
    |T_j^V(\boldsymbol t)|\le\lambda.
\]
\end{proof}

Having characterized the lasso active set through the leave-$j$-out score, we now turn to the analysis of this score. To this end, we introduce some additional notation.
For an active set $E\subseteq[p]\setminus\{j\}$ and sign vector
$\boldsymbol s\in\{-1,1\}^{|E|}$, define the polyhedral selection region
\[
    \mathcal R_{E,\boldsymbol s}^{V,-j}
    =
    \left\{
        \boldsymbol z\in\mathbb R^n:
        \operatorname{supp}\bigl(\widehat{\boldsymbol\beta}_V^{(-j)}(\boldsymbol z)\bigr)=E,\ 
        \operatorname{sign}\bigl(\widehat{\boldsymbol\beta}_{V,E}^{(-j)}(\boldsymbol z)\bigr)=\boldsymbol s
    \right\},
\]
and let $\mathcal E_{-j}^V=\left\{E:\operatorname{Leb}_n\bigl(\mathcal R_{E,\boldsymbol s}^{V,-j}\bigr)>0\text{ for some }\boldsymbol s\right\}$
be the set of essential active sets from lasso regression on $\bm{V}_{-j}$, i.e., whose corresponding selection regions have positive Lebesgue measure.

\begin{lemma}[Almost-everywhere gradient of the score]
\label{lem:ae-gradient-score}
For every $\boldsymbol z\in \cup_{E,s} \text{\normalfont int}(\mathcal R_{E,\boldsymbol s}^{V,-j})$, let
$E=\operatorname{supp}\bigl(\widehat{\boldsymbol\beta}_V^{(-j)}(\boldsymbol z)\bigr)$.
Then
\[
    \nabla T_j^V(\boldsymbol z)=\boldsymbol P_{\boldsymbol V_E}^\perp \boldsymbol V_j,
\]
where
$ \boldsymbol P_{\boldsymbol V_E}=\boldsymbol V_E(\boldsymbol V_E^\top \boldsymbol V_E)^{-1}\boldsymbol V_E^\top$, and 
    $\boldsymbol P_{\boldsymbol V_E}^\perp=\boldsymbol I-\boldsymbol P_{\boldsymbol V_E}.$
\end{lemma}

\begin{proof}
Fix $\boldsymbol z\in \cup_{E,s} \text{\normalfont int}(\mathcal R_{E,\boldsymbol s}^{V,-j})$, and let $(E,\boldsymbol s)$ denote the support and sign of
the leave-$j$-out lasso solution. 
Assumption \ref{assump:general-position} implies that $\boldsymbol V_E$ has full column rank, and therefore
\[
    \widehat{\boldsymbol\beta}_{V,E}^{(-j)}(\boldsymbol z)
    =
    (\boldsymbol V_E^\top \boldsymbol V_E)^{-1}(\boldsymbol V_E^\top \boldsymbol z-\lambda\boldsymbol s)
\]
throughout this neighborhood. 
It follows that
\[
    \boldsymbol R_{-j}^V(\boldsymbol z)
    =
    \boldsymbol z-\boldsymbol V_E\widehat{\boldsymbol\beta}_{V,E}^{(-j)}(\boldsymbol z)
    =
    \boldsymbol P_{\boldsymbol V_E}^\perp \boldsymbol z
    +
    \lambda \boldsymbol V_E(\boldsymbol V_E^\top \boldsymbol V_E)^{-1}\boldsymbol s,
\]
and therefore,
\[
    \nabla T_j^V(\boldsymbol z)
    =
    \nabla\left(\boldsymbol V_j^\top \boldsymbol R_{-j}^V(\boldsymbol z)\right)
    =
    \boldsymbol P_{\boldsymbol V_E}^\perp \boldsymbol V_j.
\]
Here, $\bm{z} \in \cup_{E,s} \text{\normalfont int}(\mathcal R_{E,\boldsymbol s}^{V,-j})$ guarantees that $\bm z$ is in the interior of $\mathcal R_{E,\boldsymbol s}^{V,-j}$, and therefore $(E, \bm s)$ stay constant in a local neighborhood of $\bm z$ during differentiation.
\end{proof}

\begin{lemma}[Sandwich-bound on leave-$j$-out score]
\label{lem:coordinate-direction-comparison}
Let $\boldsymbol v_j^V=\frac{\boldsymbol V_j}{\|\boldsymbol V_j\|_2}$ and $\boldsymbol r_{j\mid E}^V=\boldsymbol P_{\boldsymbol V_E}^\perp \boldsymbol V_j$. 
For every $\boldsymbol z\in\mathbb R^{n+1}$ and $h>0$,
\[
    T_j^V(\boldsymbol z-\rho_{ij}^V h\boldsymbol v_j^V)
    \le
    T_j^V(\boldsymbol z+h\boldsymbol e_i)
    \le
    T_j^V(\boldsymbol z+\rho_{ij}^V h\boldsymbol v_j^V),
\]
where, for $i\in[n+1]$, 
\[
    \rho_{ij}^V
    =
    \sup_{{E\in\mathcal E_{-j}^V,\boldsymbol r_{j\mid E}^V\ne0}}
    \frac{|(\boldsymbol r_{j\mid E}^V)_i|\|\boldsymbol V_j\|_2}
         {\|\boldsymbol r_{j\mid E}^V\|_2^2}.
\]
\end{lemma}

\begin{proof}
By Lemma~\ref{lem:ae-gradient-score}, for almost every
$\boldsymbol z$, if $E\in\mathcal E_{-j}^V$ is the locally active
leave-$j$-out model at $\boldsymbol z$, then
\[
    D_{\boldsymbol v_j^V}T_j^V(\boldsymbol z)
    =
    \bigl(\boldsymbol r_{j\mid E}^V\bigr)^\top
    \boldsymbol v_j^V
    =
    \frac{\|\boldsymbol r_{j\mid E}^V\|_2^2}
         {\|\boldsymbol V_j\|_2}
    \ge 0.
\]
When $\boldsymbol r_{j\mid E}^V\ne0$, the definition of
$\rho_{ij}^V$ gives
\[
\begin{aligned}
    |\partial_iT_j^V(\boldsymbol z)|
    &=
    |(\boldsymbol r_{j\mid E}^V)_i|\\
    &\le
    \rho_{ij}^V
    \frac{\|\boldsymbol r_{j\mid E}^V\|_2^2}
         {\|\boldsymbol V_j\|_2}\\
    &=
    \rho_{ij}^V
    D_{\boldsymbol v_j^V}T_j^V(\boldsymbol z).
\end{aligned}
\]
When $\boldsymbol r_{j\mid E}^V=0$, both sides vanish. Hence,
\[
    |\partial_iT_j^V|
    \le
    \rho_{ij}^V D_{\boldsymbol v_j^V}T_j^V
\]
Lebesgue-almost everywhere.

By linearity of directional derivatives in the direction of differentiation,
it follows almost everywhere that
\[
\begin{aligned}
    D_{\boldsymbol e_i+\rho_{ij}^V\boldsymbol v_j^V}T_j^V
    &=
    \partial_iT_j^V
    +
    \rho_{ij}^V D_{\boldsymbol v_j^V}T_j^V
    \ge0,\\
    D_{\rho_{ij}^V\boldsymbol v_j^V-\boldsymbol e_i}T_j^V
    &=
    \rho_{ij}^V D_{\boldsymbol v_j^V}T_j^V
    -
    \partial_iT_j^V
    \ge0.
\end{aligned}
\]

To apply Lemma~\ref{lem:ae-directional-monotonicity}, observe that
\[
\begin{aligned}
    (\boldsymbol z+h\boldsymbol e_i)
    -
    (\boldsymbol z-\rho_{ij}^Vh\boldsymbol v_j^V)
    &=
    h(\boldsymbol e_i+\rho_{ij}^V\boldsymbol v_j^V),\\
    (\boldsymbol z+\rho_{ij}^Vh\boldsymbol v_j^V)
    -
    (\boldsymbol z+h\boldsymbol e_i)
    &=
    h(\rho_{ij}^V\boldsymbol v_j^V-\boldsymbol e_i).
\end{aligned}
\]
Thus, applying the lemma first from
$\boldsymbol z-\rho_{ij}^Vh\boldsymbol v_j^V$ in direction
$\boldsymbol e_i+\rho_{ij}^V\boldsymbol v_j^V$, and then from
$\boldsymbol z+h\boldsymbol e_i$ in direction
$\rho_{ij}^V\boldsymbol v_j^V-\boldsymbol e_i$, gives
\[
    T_j^V(\boldsymbol z-\rho_{ij}^Vh\boldsymbol v_j^V)
    \le
    T_j^V(\boldsymbol z+h\boldsymbol e_i)
    \le
    T_j^V(\boldsymbol z+\rho_{ij}^Vh\boldsymbol v_j^V).
\]
This proves the claim.
\end{proof}

\subsubsection{Stability bound for feature--selection probabilities}

\begin{lemma}[Bounding probabilities under randomization]
\label{lem:gaussian-score-bound}
Let $A\subseteq\mathbb R^{n+1}$ be measurable and define
\[
    p_A(\boldsymbol z)=\mathbb P_{\boldsymbol\omega}(\boldsymbol z+\boldsymbol\omega\in A),
    \qquad
    \boldsymbol\omega\sim N(0,\tau^2\boldsymbol I_n).
\]
Then, for every unit vector $\boldsymbol v$, $D_{\boldsymbol v}p_A(\boldsymbol z)
=\frac1{\tau^2}\mathbb E_{\boldsymbol\omega}\left[ (\boldsymbol v^\top \boldsymbol\omega)\mathds{1}\{\boldsymbol z+\boldsymbol\omega\in A\}\right]$, 
and $|D_{\boldsymbol v}p_A(\boldsymbol z)|\le \frac1{\tau\sqrt{2\pi}}$.
\end{lemma}

\begin{proof}
Differentiation under the Gaussian convolution gives us 
\[D_{\boldsymbol v}p_A(\boldsymbol z)
=\frac1{\tau^2}\mathbb E_{\boldsymbol\omega}\left[ (\boldsymbol v^\top \boldsymbol\omega)\mathds{1}\{\boldsymbol z+\boldsymbol\omega\in A\}\right].
\]

If $G=\boldsymbol v^\top\boldsymbol\omega$, then
$G\sim N(0,\tau^2)$ and
\[
    -\mathbb E_{\boldsymbol\omega}[(-G)_+]
    \le
    \mathbb E_{\boldsymbol\omega}\left[G\mathds{1}\{\boldsymbol z+\boldsymbol\omega\in A\}\right]
    \le
    \mathbb E_{\boldsymbol\omega}[G_+], \text{ where } G_+= \max\{G, 0\}.
\]
By symmetry, $\mathbb E_{\boldsymbol\omega}[G_+]=\mathbb E_{\boldsymbol\omega}[(-G)_+]=\frac{\tau}{\sqrt{2\pi}}$, and the result follows.
\end{proof}

\begin{theorem}[Sensitivity bound for selection probabilities under randomized lasso]
\label{thm:single-feature-stability}
Under Assumption {\ref{assump:general-position}}, for every $\boldsymbol z\in\mathbb R^{n+1}$ and
$i\in[n+1]$,
\begin{equation}
    |\partial_i\pi_j^{\circ}(\boldsymbol z)|
    \le
    \sqrt{\frac{2}{\pi}}\frac{\rho_{ij}^V}{\tau}.
    \label{eq:single-feature-coordinate-bound}
\end{equation}
Consequently, we have 
\begin{equation}
    \sup_{\boldsymbol z\in\mathbb R^n}
    \|\nabla\pi_j^{\circ}(\boldsymbol z)\|_\infty
    \le
    \sqrt{\frac{2}{\pi}}\frac{a_j^V}{\tau},
    \label{eq:single-feature-gradient-bound}
\end{equation}
where 
\[a_j^V
    :=
    \max_{i\in[n+1]}\rho_{ij}^V
    =
    \sup_{{E\in\mathcal E_{-j}^V,\boldsymbol r_{j\mid E}^V\ne0}}
    \frac{\|\boldsymbol r_{j\mid E}^V\|_\infty\|\boldsymbol V_j\|_2}
         {\|\boldsymbol r_{j\mid E}^V\|_2^2}.
\]
\end{theorem}

\begin{proof}
By Lemma~\ref{lem:leave-j-out-kkt}, 
\[
    \pi_j^{\circ}(\boldsymbol z)
    =
    p_{j,+}(\boldsymbol z)+p_{j,-}(\boldsymbol z),
\]
where
\[
    p_{j,+}(\boldsymbol z)
    =
    \mathbb P_{\boldsymbol\omega}\{T_j^V(\boldsymbol z+\boldsymbol\omega)>\lambda\},
    \qquad
    p_{j,-}(\boldsymbol z)
    =
    \mathbb P_{\boldsymbol\omega}\{T_j^V(\boldsymbol z+\boldsymbol\omega)<-\lambda\}.
\]

Then, applying Lemma \ref{lem:coordinate-direction-comparison} yields
\[
    p_{j,+}(\boldsymbol z-\rho_{ij}^V h\boldsymbol v_j^V)
    \le
    p_{j,+}(\boldsymbol z+h\boldsymbol e_i)
    \le
    p_{j,+}(\boldsymbol z+\rho_{ij}^V h\boldsymbol v_j^V).
\]
Letting $h\downarrow0$, we have
\[
    |\partial_i p_{j,+}(\boldsymbol z)|
    \le
    \rho_{ij}^V|D_{\boldsymbol v_j^V}p_{j,+}(\boldsymbol z)|.
\]
Analogously, applying Lemma \ref{lem:coordinate-direction-comparison} to $p_{j,-}(\boldsymbol z)$ yields
\[
|\partial_i p_{j,-}(\boldsymbol z)|
    \le
    \rho_{ij}^V|D_{\boldsymbol v_j^V}p_{j,-}(\boldsymbol z)|
\]

Finally, applying Lemma~\ref{lem:gaussian-score-bound} to bound the directional derivatives on the right-hand side directional derivatives of the two above-stated inequalities, we have
\[
    |\partial_i p_{j,+}(\boldsymbol z)|
    \le
    \frac{\rho_{ij}^V}{\tau\sqrt{2\pi}},
    \qquad
    |\partial_i p_{j,-}(\boldsymbol z)|
    \le
    \frac{\rho_{ij}^V}{\tau\sqrt{2\pi}}.
\]
Therefore,
\[
    |\partial_i\pi_j^{\circ}(\boldsymbol z)|
    \le
    \frac{2\rho_{ij}^V}{\tau\sqrt{2\pi}}
    =
    \sqrt{\frac{2}{\pi}}\frac{\rho_{ij}^V}{\tau}.
\]
Taking the maximum over $i$ proves
\eqref{eq:single-feature-gradient-bound}.
\end{proof}

\subsection{Auxiliary results}
\label{App:stability:aux}

For $E\in\mathcal E_{-j}^V$, recall from Lemma \ref{lem:ae-gradient-score} and Lemma \ref{lem:coordinate-direction-comparison} that $\boldsymbol P_{\boldsymbol V_E}$ denotes the orthogonal projection onto the column space of $\boldsymbol V_E$, and that 
\[
    \boldsymbol r_{j\mid E}^V
    =
    \boldsymbol P_{\boldsymbol V_E}^{\perp}\boldsymbol V_j,
    \quad \text{where, } \quad 
    \boldsymbol P_{\boldsymbol V_E}^{\perp}
    =
    \boldsymbol I_{n+1}-\boldsymbol P_{\boldsymbol V_E},
\]
and that for $i \in [n+1]$,
\[
    \rho_{ij}^V
    =
    \sup_{\substack{
        E\in\mathcal E_{-j}^V\\
        \boldsymbol r_{j\mid E}^V\ne0
    }}
    \frac{
        |(\boldsymbol r_{j\mid E}^V)_i|\,\|\boldsymbol V_j\|_2
    }{
        \|\boldsymbol r_{j\mid E}^V\|_2^2
    }.
\]

\begin{proposition}[Sensitivity bound under randomized local distillation]
\label{thm:distillation-coordinate-stability}
Suppose that $\boldsymbol y\longmapsto\hat\phi(\bm{x}^*, \bm{y})$
belongs to $C^1(\mathbb R^n)$, and that Assumption {\ref{assump:general-position}} holds. Then, for every $\boldsymbol y\in\mathbb R^n$,
$i\in[n]$, and $j\in[p]$,
\begin{align*}
    |\partial_i\pi_j(\boldsymbol y)|
    \le
    \sqrt{\frac{2}{\pi}}\frac{1}{\tau}
    \left\{
        \sqrt{S_i}\,\rho_{ij}^V
        +
        \sqrt{S_{n+1}}\,
        |\partial_i\hat\phi(\bm{x}^*, \bm{y})|\,
        \rho_{n+1,j}^V
    \right\}.
\end{align*}
\end{proposition}

\begin{proof}
By Lemma~\ref{lem:smoothness},
\[
    \pi_j(\boldsymbol y)
    =
    \pi_j^\circ\bigl(\boldsymbol z(\boldsymbol y)\bigr)
\]
is continuously differentiable. 

The coordinates of $\boldsymbol z(\boldsymbol y)$ satisfy
\[
    z_k(\boldsymbol y)
    =
    \sqrt{S_k}\,y_k,
    \qquad k\in[n], \qquad
    z_{n+1}(\boldsymbol y)
    =
    \sqrt{S_{n+1}}\,\hat\phi(\bm{x}^*, \bm{y}),
\]
where $S_{n+1}=\frac{\mu}{\sqrt{n_{\mathrm{eff}}}}$.
Consequently, for $i\in[n]$ and $k\in[n+1]$,
\[
    \frac{\partial z_k(\boldsymbol y)}{\partial y_i}
    =
    \sqrt{S_i}\,\mathds{1}\{k=i\}
    +
    \sqrt{S_{n+1}}\,
    \partial_i\hat\phi(\bm{x}^*, \bm{y})\,
    \mathds{1}\{k=n+1\}.
\]
The chain rule therefore gives
\begin{align}
    \partial_i\pi_j(\boldsymbol y)
    &=
    \sum_{k=1}^{n+1}
    \partial_{z_k}\pi_j^\circ\bigl(\boldsymbol z(\boldsymbol y)\bigr)
    \frac{\partial z_k(\boldsymbol y)}{\partial y_i}
    \nonumber\\
    &=
    \sqrt{S_i}\,
    \partial_{z_i}\pi_j^\circ\bigl(\boldsymbol z(\boldsymbol y)\bigr)
    +
    \sqrt{S_{n+1}}\,
    \partial_i\hat\phi(\bm{x}^*, \bm{y})\,
    \partial_{z_{n+1}}\pi_j^\circ\bigl(\boldsymbol z(\boldsymbol y)\bigr).
    \label{eq:distillation-chain-rule}
\end{align}

Applying Theorem~\ref{thm:single-feature-stability} gives, for every $\boldsymbol z\in\mathbb R^{n+1}$ and $k\in[n+1]$,
\[
    \left|
        \partial_{z_k}\pi_j^\circ(\boldsymbol z)
    \right|
    \le
    \sqrt{\frac{2}{\pi}}\frac{\rho_{kj}^V}{\tau}.
\]
Substituting these bounds into
\eqref{eq:distillation-chain-rule} and applying the triangle
inequality proves the bound.
\end{proof}

\begin{lemma}[Directional monotonicity]
\label{lem:ae-directional-monotonicity}
Let $f:\mathbb R^n\to\mathbb R$ be Lipschitz, and let
$\boldsymbol d\in\mathbb R^n$. 
If $D_{\boldsymbol d}f(\boldsymbol z)\ge0$ for Lebesgue-almost every $\boldsymbol z$, then $f(\boldsymbol z+t\boldsymbol d)\ge f(\boldsymbol z)$ for every $\boldsymbol z\in\mathbb R^n$ and $t\ge0$.
\end{lemma}

\begin{proof}
The result is immediate if $\boldsymbol d=0$. Otherwise, let
$\widetilde{\boldsymbol d}=\boldsymbol d/\|\boldsymbol d\|_2$ and decompose
$\mathbb R^n=\widetilde{\boldsymbol d}^\perp\oplus\operatorname{span}(\widetilde{\boldsymbol d})$, and for $\boldsymbol u\in\widetilde{\boldsymbol d}^\perp$, define $g_{\boldsymbol u}(t)=f(\boldsymbol u+t\widetilde{\boldsymbol d})$.
By Fubini's theorem, for almost every $\boldsymbol u$, the inequality
\[
    D_{\widetilde{\boldsymbol d}}f(\boldsymbol u+t\widetilde{\boldsymbol d})\ge0
\]
holds for almost every $t\in\mathbb R$. Since $f$ is Lipschitz,
$g_{\boldsymbol u}$ is absolutely continuous, and
\[
    g_{\boldsymbol u}'(t)
    =
    D_{\widetilde{\boldsymbol d}}f(\boldsymbol u+t\widetilde{\boldsymbol d})
    \ge0
\]
for almost every $t$. Hence $g_{\boldsymbol u}$ is nondecreasing.

Now fix an arbitrary $\boldsymbol u\in\widetilde{\boldsymbol d}^\perp$. Choose a sequence
$\boldsymbol u_m\to \boldsymbol u$ such that $g_{\boldsymbol u_m}$ is nondecreasing. For $t_1<t_2$,
\[
    f(\boldsymbol u_m+t_1\widetilde{\boldsymbol d})
    \le
    f(\boldsymbol u_m+t_2\widetilde{\boldsymbol d}).
\]
Passing to the limit and using continuity of $f$ yields $f(\boldsymbol u+t_1\widetilde{\boldsymbol d})\le f(\boldsymbol u+t_2\widetilde{\boldsymbol d})$.
Thus $f$ is nondecreasing along every line parallel to $\boldsymbol d$.
\end{proof}

\begin{lemma}[Effective similarity neighborhood: size and weights]
\label{lem:distillation-effective-neighborhood-size}
Let $\mathcal I_S=\left\{i\in[n]:S_i\ge\frac{1}{2n_{\mathrm{eff}}}\right\}$.
Under Assumption~\ref{assump:distillation-weight-spikiness}, $|\mathcal I_S| \ge \frac{n_{\mathrm{eff}}}{2C_S^2}$.
Moreover, for every $i\in\mathcal I_S$, $\frac{1}{2n_{\mathrm{eff}}}\le S_i\le
\frac{C_S}{n_{\mathrm{eff}}}$.
\end{lemma}

\begin{proof}
For every $i\notin\mathcal I_S$, the definition of
$\mathcal I_S$ gives $S_i<\frac{1}{2n_{\mathrm{eff}}}$, and therefore
\[
    \sum_{i\notin\mathcal I_S}S_i^2
    \le
    \frac{1}{2n_{\mathrm{eff}}}
    \sum_{i\notin\mathcal I_S}S_i
    \le
    \frac{1}{2n_{\mathrm{eff}}},
\]
where the last inequality uses
$\sum_{i=1}^nS_i=1$. 
Since $\sum_{i=1}^nS_i^2=\frac{1}{n_{\mathrm{eff}}}$,
it follows that
\[
    \sum_{i\in\mathcal I_S}S_i^2
    \ge
    \frac{1}{2n_{\mathrm{eff}}}.
\]

On the other hand, Assumption~\ref{assump:distillation-weight-spikiness} implies
$S_i\le S_{\max}\le\frac{C_S}{n_{\mathrm{eff}}}$, for $i\in[n]$.
Therefore, 
$$\sum_{i\in\mathcal I_S}S_i^2\le |\mathcal I_S| \frac{C_S^2}{n_{\mathrm{eff}}^2}.$$

Combining the preceding two displays yields $|\mathcal I_S|
    \frac{C_S^2}{n_{\mathrm{eff}}^2}
    \ge
    \frac{1}{2n_{\mathrm{eff}}}$, i.e.,
\[
    |\mathcal I_S|
    \ge
    \frac{n_{\mathrm{eff}}}{2C_S^2}.
\]

Finally, for $i\in\mathcal I_S$, the lower bound on $S_i$ follows
from the definition of $\mathcal I_S$, while the upper bound follows
from
Assumption~\ref{assump:distillation-weight-spikiness}.
\end{proof}

\begin{lemma}[Bounds on residual and design in augmented regression]
\label{lem:distillation-augmented-residual-geometry}
Under
Assumptions~\ref{assump:distillation-weight-spikiness}--%
\ref{assump:distillation-local-nondegeneracy}, define
$C_V =B_X\sqrt{1+C_\mu}$, and $c_R=\frac{\kappa_X}{2C_S}$.
Then, uniformly over $j\in[p]$,
\[
    \|\boldsymbol V_j\|_2\le C_V,
\]
and, uniformly over $j\in[p]$ and
$E\in\mathcal E_{-j}^V$,
\[
    \|\boldsymbol r_{j\mid E}^V\|_2\ge c_R.
\]
In particular, $\boldsymbol r_{j\mid E}^V\ne0$ for every $j\in[p]$ and every
$E\in\mathcal E_{-j}^V$.
\end{lemma}

\begin{proof}
First, because the training weights are nonnegative and sum to one,
\[
    \sum_{i=1}^nS_i^2
    \le
    \left(\sum_{i=1}^nS_i\right)^2
    =
    1.
\]
Therefore, \(n_{\mathrm{eff}}\ge1\). Using
\(S_{n+1}=\mu/\sqrt{n_{\mathrm{eff}}}\) and
Assumption~\ref{assump:distillation-bounded-covariates}, we obtain
\begin{align*}
    \|\boldsymbol V_j\|_2^2
    &=
    \sum_{i=1}^nS_i x_{ij}^2
    +
    S_{n+1}(x_j^*)^2
    \le
    B_X^2\sum_{i=1}^nS_i
    +
    \frac{\mu}{\sqrt{n_{\mathrm{eff}}}}B_X^2
    \le
    B_X^2(1+C_\mu)
    =
    C_V^2.
\end{align*}
This proves the asserted upper bound on \(\|\boldsymbol V_j\|_2\).

Next, fix arbitrary $j\in[p]$ and \(E\in\mathcal E_{-j}^V\). 
By the definition of an orthogonal projection,
$\|\boldsymbol r_{j\mid E}^V\|_2^2
    =
    \inf_{\boldsymbol{\gamma}\in\mathbb R^{|E|}}
    \|\boldsymbol V_j-\boldsymbol V_E\boldsymbol{\gamma}\|_2^2$.
For every \(\boldsymbol{\gamma}\in\mathbb R^{|E|}\),
\begin{align*}
    \|\boldsymbol V_j-\boldsymbol V_E\boldsymbol{\gamma}\|_2^2
    &=
    \sum_{i=1}^n
        S_i\bigl(x_{ij}-\boldsymbol x_{i,E}^{\top}\boldsymbol{\gamma}\bigr)^2
    +
    S_{n+1}
        \bigl(x_j^*-\boldsymbol x_E^{*\top}\boldsymbol{\gamma}\bigr)^2\\
    &\ge
    \sum_{i\in\mathcal I_S}
        S_i\bigl(x_{ij}-\boldsymbol x_{i,E}^{\top}\boldsymbol{\gamma}\bigr)^2\ge
    \frac{1}{2n_{\mathrm{eff}}}
    \left\|
        \boldsymbol X_{\mathcal I_S,j}
        -
        \boldsymbol X_{\mathcal I_S,E}\boldsymbol{\gamma}
    \right\|_2^2,
\end{align*}
where the final inequality follows from the definition of
\(\mathcal I_S\).

After ordering the columns of
\(\boldsymbol X_{\mathcal I_S,E\cup\{j\}}\) so that \(j\) is last, we write
\[
    \boldsymbol X_{\mathcal I_S,j}
    -
    \boldsymbol X_{\mathcal I_S,E}\boldsymbol{\gamma}
    =
    \boldsymbol X_{\mathcal I_S,E\cup\{j\}}
    \begin{pmatrix}
        -\boldsymbol{\gamma}\\
        1
    \end{pmatrix}.
\]
Assumption~\ref{assump:distillation-local-nondegeneracy} therefore
implies
\begin{align*}
    \left\|
        \boldsymbol X_{\mathcal I_S,j}
        -
        \boldsymbol X_{\mathcal I_S,E}\boldsymbol{\gamma}
    \right\|_2
    &\ge
    \kappa_X\sqrt{|\mathcal I_S|}
    \left\|
        \begin{pmatrix}
            -\boldsymbol{\gamma}\\
            1
        \end{pmatrix}
    \right\|_2
    \ge
    \kappa_X\sqrt{|\mathcal I_S|}.
\end{align*}
It follows that
\[
    \|\boldsymbol V_j-\boldsymbol V_E\boldsymbol{\gamma}\|_2^2
    \ge
    \frac{\kappa_X^2|\mathcal I_S|}
         {2n_{\mathrm{eff}}}.
\]
By Lemma~\ref{lem:distillation-effective-neighborhood-size}, $|\mathcal I_S|\ge\frac{n_{\mathrm{eff}}}{2C_S^2}$, and therefore,
\[
    \|\boldsymbol V_j-\boldsymbol V_E\boldsymbol{\gamma}\|_2^2
    \ge
    \frac{\kappa_X^2}{4C_S^2}
    =
    c_R^2.
\]
Taking the infimum over \(\boldsymbol{\gamma}\) proves $\|\boldsymbol r_{j\mid E}^V\|_2\ge c_R$, which completes the proof.
\end{proof}

\begin{lemma}[Bounds for weighted contributions from local distillation]
\label{lem:distillation-native-comparison-bounds}
Under Assumptions~\ref{assump:distillation-weight-spikiness}--%
\ref{assump:distillation-local-nondegeneracy}, there exists a constant
$C_\rho<\infty$, such that, uniformly over $j\in[p]$,
\[
    \max_{i\in[n]}
    \sqrt{S_i}\,\rho_{ij}^V
    \le
    \frac{C_\rho}{\sqrt{n_{\mathrm{eff}}}}, \quad
    \sqrt{S_{n+1}}\,\rho_{n+1,j}^V
    \le
    \frac{C_\rho}{n_{\mathrm{eff}}^{1/4}}.
\]
\end{lemma}

\begin{proof}
By Lemma~\ref{lem:distillation-augmented-residual-geometry},
uniformly over $j\in[p]$ and $E\in\mathcal E_{-j}^V$, we have
\[
    \|\boldsymbol V_j\|_2\le C_V,
    \qquad
    \|\boldsymbol r_{j\mid E}^V\|_2\ge c_R.
\]
Since  $|(\boldsymbol r_{j\mid E}^V)_k|\le \|\boldsymbol r_{j\mid E}^V\|_2$,
we obtain, for every $k\in[n+1]$,
\[
    \rho_{kj}^V
    \le
    \frac{C_V}{c_R}
    =:\overline C_\rho.
\]

For $i\in[n]$, Assumption~%
\ref{assump:distillation-weight-spikiness} gives $S_i\le S_{\max}\le \frac{C_S}{n_{\mathrm{eff}}}$, and therefore
\[
    \sqrt{S_i}\,\rho_{ij}^V
    \le
    \frac{\sqrt{C_S}\,\overline C_\rho}
         {\sqrt{n_{\mathrm{eff}}}}.
\]
Finally,
\[
    \sqrt{S_{n+1}}\,\rho_{n+1,j}^V
    =
    \frac{\sqrt{\mu}}{n_{\mathrm{eff}}^{1/4}}
    \rho_{n+1,j}^V
    \le
    \frac{\sqrt{C_\mu}\,\overline C_\rho}
         {n_{\mathrm{eff}}^{1/4}}.
\]
Letting $C_\rho=\sqrt{C_S}\,\overline C_\rho \vee \sqrt{C_\mu}\,\overline C_\rho$ proves both claims.
\end{proof}

\section{Performance on UCI ML and OpenML datasets}
\label{app:real_data}

Here we report the results across all 17 benchmark datasets and methods, including a second tabular foundation-model teacher, TabFM~\citep{kong2026tabfm}, and a cross-validated teacher-selection rule, which selects, per split, the teacher with the lowest cross-validated error. Table~\ref{tab:datasets} lists per-dataset sample sizes and feature counts; Figure~\ref{fig:appendix_r2} shows test $R^2$.

  \begin{table}[ht] 
  \centering      
  \footnotesize
  \begin{tabular}{lrrl}
  \toprule
  Dataset & $n$ & $p$ & Source \\
  \midrule 
   Automobile & 159 & 51 & UCI \\
  Servo & 167 & 10 & UCI \\
  Liver Disorders & 341 & 5 & UCI \\
  Auto MPG & 392 & 8 & UCI \\
  Real Estate Valuation & 414 & 6 & UCI \\
  Student Performance & 649 & 39 & UCI \\
  Cars & 804 & 17 & OpenML \\
  QSAR Fish Toxicity & 907 & 6 & OpenML \\
  Concrete Compressive Strength & 1005 & 8 & OpenML \\
  Infrared Thermography Temperature & 1018 & 43 & UCI \\
  Socmob & 1156 & 35 & OpenML \\
  Red Wine & 1359 & 11 & OpenML \\
  Airfoil Self-Noise & 1503 & 5 & OpenML \\
  Auction Verification & 2043 & 7 & OpenML \\
  Space GA & 3107 & 6 & OpenML \\
  White Wine & 3961 & 11 & OpenML \\
  Abalone & 4177 & 9 & OpenML \\
  \bottomrule
  \end{tabular}
\caption{\textbf{Benchmark datasets after complete-case filtering and de-duplication, sorted by sample size.} $p$ is the number of columns in the design matrix, i.e.\ after one-hot encoding of categorical variables.}  
\label{tab:datasets}    
  \end{table}

\clearpage

\begin{figure}[!t]
    \centering 
    \includegraphics[width=\linewidth]{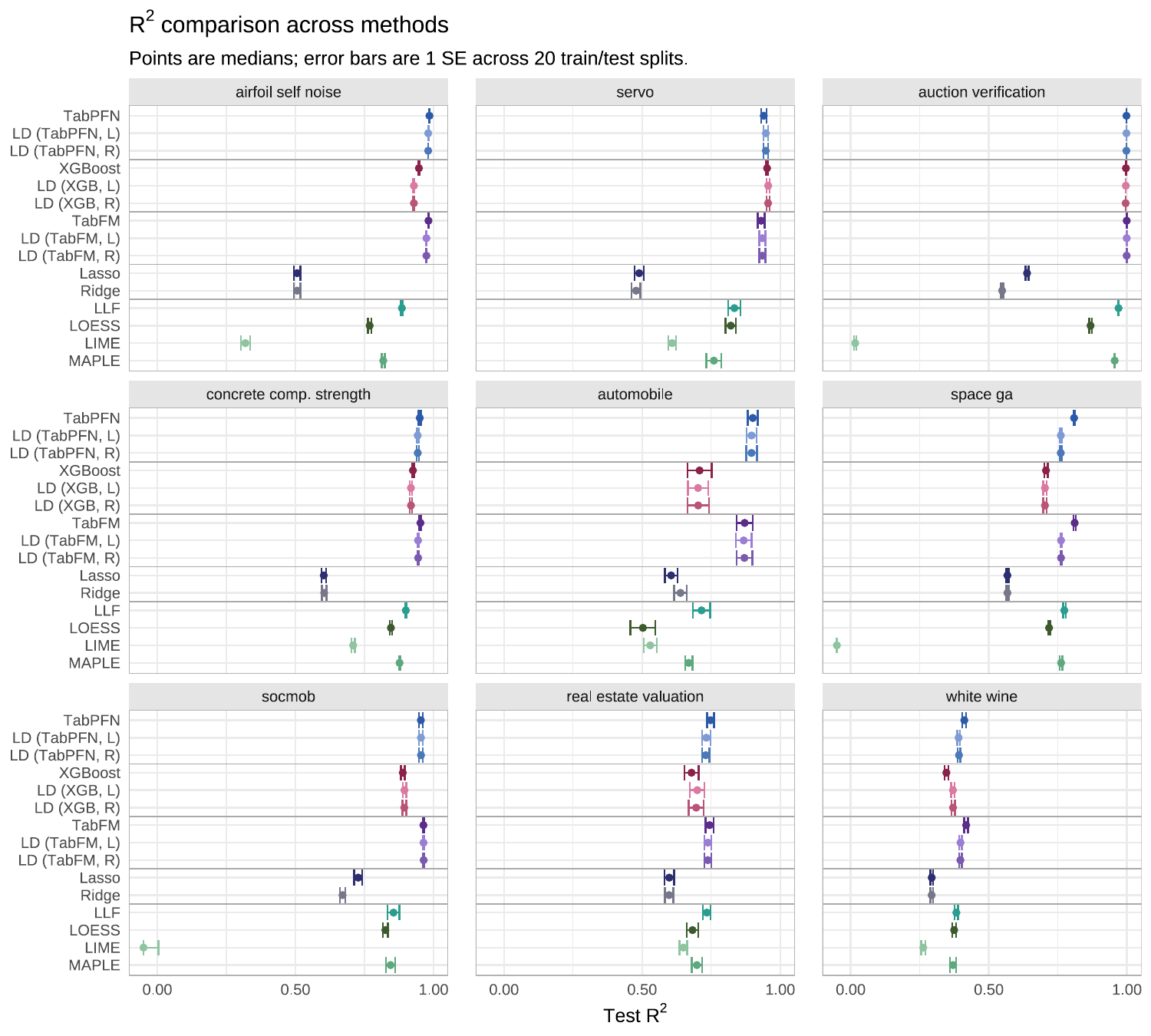} 
    \caption{\textbf{Predictive performance (test $R^2$) across 17 UCI ML and OpenML regression datasets.} Each teacher is shown with its two distilled local-linear students, lasso~(L) and ridge~(R). Local distillation approaches its teacher's accuracy across datasets. \emph{Dist (best teacher)} selects, per split, the teacher with the lowest cross-validated error. LLF and LOESS are teacher-unaware local methods and are more variable.} 
    \label{fig:appendix_r2}
\end{figure} 

\begin{figure}[!t]\ContinuedFloat
    \centering 
    \includegraphics[width=\linewidth]{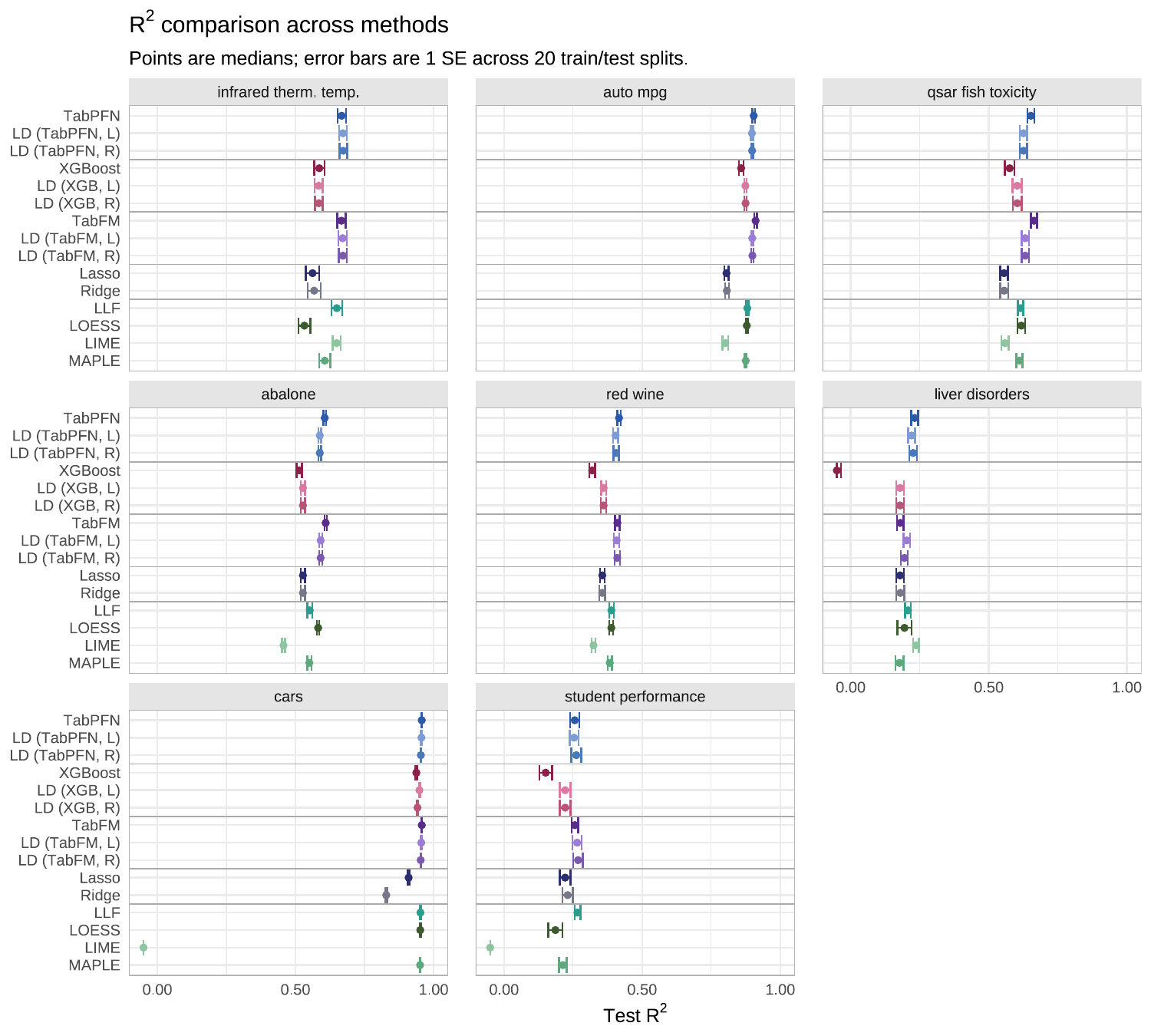} 
  \caption{(Continued.) Remaining datasets, same methods and axes as above.}
\end{figure} 

\section{Ablation study}
\label{app:ablation}

Local distillation has two components: (1) the similarity weights and (2) the ``prior'' prediction from the teacher. Here, we find that both combined usually have the best predictive performance (Figure~\ref{fig:appendix_ablation}).

\begin{figure}[!t]
    \centering 
    \includegraphics[width=\linewidth]{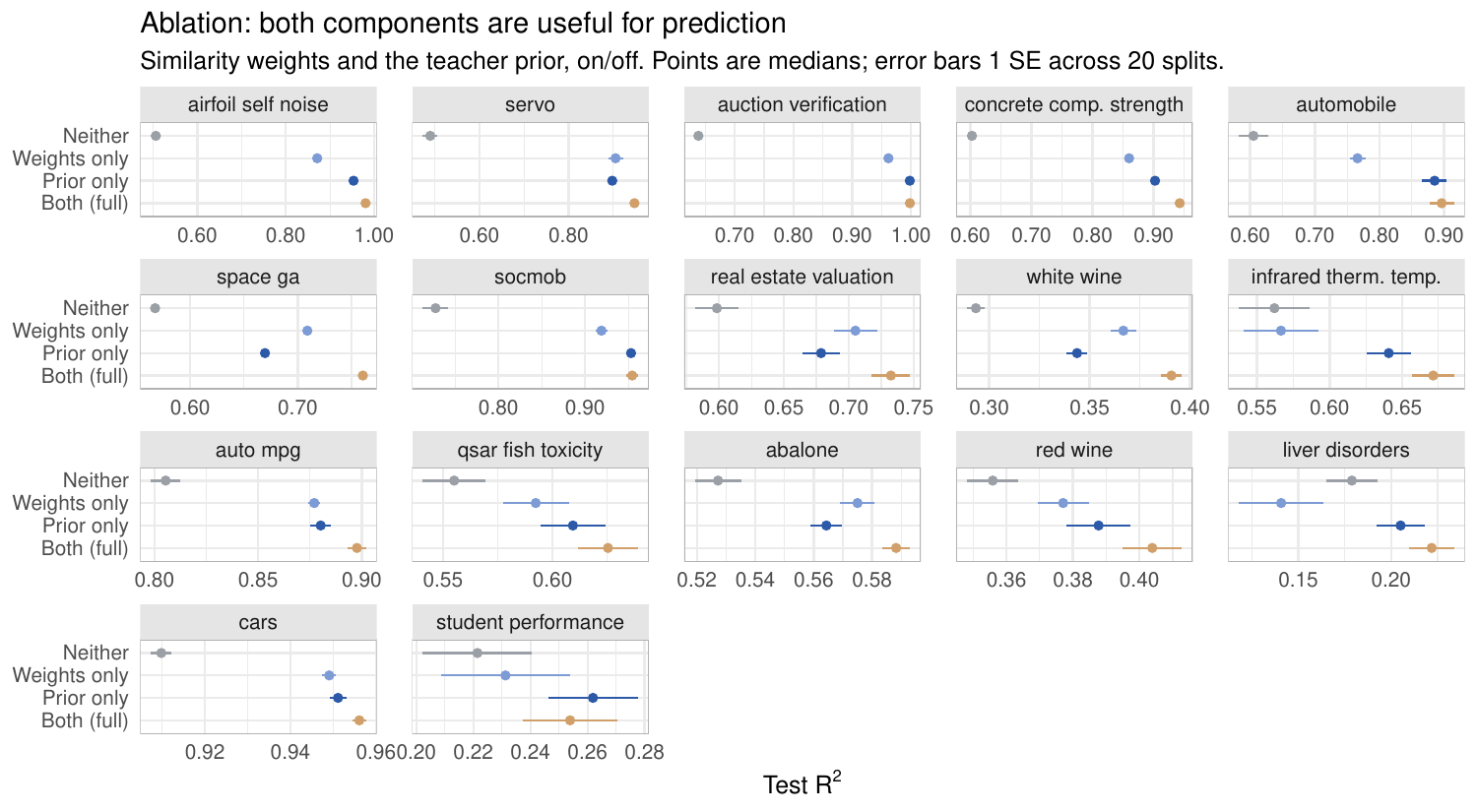} 
    \caption{\textbf{Similarity weights and prediction prior are both useful for prediction.}  Ablation of local distillation's two teacher-derived components: the \emph{similarity weights} and the \emph{prediction prior}. There are four comparators: \emph{Neither} (a global lasso), \emph{Weights only}, \emph{Prior only}, and \emph{Both} (the full method); each panel shows one dataset. Points show median test $R^2$ over 20 splits ($\pm1$ SE); the teacher is TabPFN, the student a lasso. Both components help, and using both is best or tied-best on nearly every dataset.}
    \label{fig:appendix_ablation}
\end{figure}

\section{Performance with simulated data}
\label{app:simulation}

Local distillation is intended to summarize the feature--outcome relationship linearly at each prediction point.
Here we study its performance when the data generating process is known, and test how the method responds to changes in the amount of noise, and the sizes of $n$ and $p$.
 We simulate training data as follows:
\begin{align}
 x_i &\sim N_p(0, \Sigma), \qquad \Sigma_{jk} = 0.3^{|j-k|}, \nonumber \\
 y_i &= f(x_i) + \varepsilon_i, \qquad \varepsilon_i \sim N(0, \sigma^2),
 \quad \text{where} \nonumber \\
 f(x) &= (1 + x_3)\,x_1 + (1 - x_3)\,x_2 + 2x_3 \nonumber \\
      &= x_1 + x_2 + (2 + x_1  - x_2) x_3.
 \label{eq:sim-dgp}
\end{align}
The \textit{true} local coefficients
$\nabla f(x) = (1 + x_3,\ 1 - x_3,\ 2 + x_1 - x_2,\ 0, \ldots, 0)$ vary with respect to $x_1, x_2$ and $x_3$, but the global linear model cannot model this heterogeneity. Our testing data in all cases is $40$ query points drawn from the same distribution.

We vary $n$, $p$ and $\sigma$; for each combination, we run 50 iterations of local distillation with a gradient boosting teacher and lasso student.
Table~\ref{tab:sim-recovery} reports (1) correlations across query points between the fitted local coefficients $\hat\beta_1, \hat\beta_2$ and $\hat\beta_3$ and the true local gradients; (2) support recovery described by feature selection true positive rate and the number of false positives (incorrectly selected features); (3) $\hat\mu$ and the fraction of runs in which $\hat\mu \le 1$ and the method reverted to the global fit.

We find that local distillation is generally strong across $\sigma$, $n$ and $p$ (including $p > n$), though it naturally degrades as the level of noise grows (and eventually reverts to the global linear model). In particular, the first two coefficients $(1 + x_3)$ and $(1 - x_3)$ are typically well fit, with correlation to the true values between $0.63$ and $0.79$. 
However, the third coefficient, $2 + x_1 - x_2$, is harder to recover (correlation $0.14$--$0.34$). This is a result of our definition of ``similarity'': the difference $x_1 - x_2$ affects the response only through its interaction with $x_3$, so observations with very different values of $x_1 - x_2$ can have similar teacher predictions, and each local fit returns roughly the average coefficient over its neighborhood.  

\begin{table}[t]
\centering
\setlength{\tabcolsep}{3.5pt}
{\footnotesize
\begin{tabular}{rrr | cc | ccc | cc}
\toprule
$n$ & $p$ & $\sigma$ & $\hat\mu$ & reverted
  & $\mathrm{corr}(\hat\beta_1)$ & $\mathrm{corr}(\hat\beta_2)$
  & $\mathrm{corr}(\hat\beta_3)$ & TPR & FP \\
\midrule
400 & 100 & 0.3 & 2.89 (0.49) & 0/50 & 0.79 (0.06) & 0.78 (0.06) & 0.34 (0.13) & 0.97 (0.02) & 7.4 (3.7) \\
400 & 100 & 1 & 1.46 (0.15) & 0/50 & 0.77 (0.07) & 0.77 (0.07) & 0.30 (0.13) & 0.96 (0.02) & 9.3 (5.6) \\
400 & 100 & 2 & 1.07 (0.06) & 11/50 & 0.73 (0.08) & 0.72 (0.08) & 0.24 (0.16) & 0.94 (0.04) & 10.3 (5.9) \\
400 & 100 & 3 & 1.02 (0.02) & 40/50 & 0.63 (0.09) & 0.64 (0.06) & 0.14 (0.14) & 0.92 (0.05) & 9.3 (5.3) \\
200 & 100 & 0.3 & 1.86 (0.37) & 0/50 & 0.76 (0.07) & 0.75 (0.07) & 0.33 (0.17) & 0.96 (0.02) & 7.7 (4.6) \\
400 & 800 & 0.3 & 2.43 (0.33) & 0/50 & 0.78 (0.06) & 0.77 (0.07) & 0.34 (0.15) & 0.95 (0.02) & 10.4 (7.6) \\
200 & 500 & 0.3 & 1.66 (0.36) & 0/50 & 0.77 (0.06) & 0.75 (0.08) & 0.25 (0.17) & 0.94 (0.03) & 9.8 (7.8) \\
\bottomrule
\end{tabular}
}
\caption{\textbf{Recovery of local structure using a gradient boosting teacher and lasso student.} Means (standard
deviations) reported over runs, with $50$ replicates per row. ``reverted'' counts replicates with $\hat\mu \le 1$, in which the method reverted to the global fit. The correlation columns show the correlations between the fitted and true coefficients, for runs when $\hat\mu > 1$. ``TPR'' is the true positive rate of the selected coefficients, ``FP'' is the number of false positive selections.}
\label{tab:sim-recovery}
\end{table}

We then repeated this experiment to test the effect of filtering by selection probabilities. For each configuration, we took the first $10$ replicates, selected $\hat t$ using the rule of Section~\ref{sec:randomized_estimator}, and ran $B = 100$ randomized refits per query point to compute $\hat\pi_j$ for every selected feature. Table~\ref{tab:sim-stability} reports the results. We find that filtering reduces the false positive rate for feature selection, and retains most of the true positives.

\begin{table}[t]
\centering
\setlength{\tabcolsep}{3.5pt}
{\footnotesize
\begin{tabular}{rrr c c | cc | cc | c}
\toprule
&&&&& \multicolumn{2}{c|}{share with $\hat\pi_j > 0.9$}
   & \multicolumn{2}{c|}{FP per query} & \\
$n$ & $p$ & $\sigma$ & reverted & $\hat t$
  & active & spurious & unfiltered & filtered & TPR (filtered) \\
\midrule
400 & 100 & 0.3 & 0/10 & 1.14 (0.98) & 0.98 (0.02) & 0.12 (0.14) & 8.3 (5.1) & 0.9 (1.0) & 0.95 (0.02) \\
400 & 100 & 1 & 0/10 & 0.67 (0.76) & 0.98 (0.03) & 0.21 (0.23) & 10.4 (6.5) & 2.0 (2.2) & 0.93 (0.05) \\
400 & 100 & 2 & 0/10 & 0.89 (0.92) & 0.95 (0.06) & 0.12 (0.12) & 12.4 (6.7) & 1.8 (2.2) & 0.87 (0.09) \\
400 & 100 & 3 & 8/10 & 0.62 (0.53) & 0.93 (0.03) & 0.05 (0.07) & 9.1 (0.6) & 0.5 (0.7) & 0.86 (0.01) \\
200 & 100 & 0.3 & 0/10 & 0.39 (0.37) & 0.98 (0.02) & 0.21 (0.27) & 7.3 (4.9) & 1.6 (2.0) & 0.94 (0.03) \\
400 & 800 & 0.3 & 0/10 & 0.66 (0.76) & 0.98 (0.02) & 0.15 (0.22) & 10.1 (9.1) & 1.7 (2.3) & 0.93 (0.02) \\
200 & 500 & 0.3 & 0/10 & 0.42 (0.75) & 0.98 (0.02) & 0.27 (0.28) & 9.0 (5.5) & 2.4 (2.9) & 0.94 (0.01) \\
\bottomrule
\end{tabular}
}
\caption{\textbf{Validation of the selection probabilities against ground truth.} Means (standard deviations) across replicates with
$\hat\mu > 1$; maximum $10$ per row, with $B = 100$ randomized refits per query point at the scale $\hat t$ selected per replicate as described in Section~\ref{sec:randomized_estimator}. ``Reverted'' counts replicates with $\hat\mu \le 1$, which are excluded from the other columns. 
``Active'' and ``spurious'' refer to selections of the $3$ truly active and the $p - 3$ null features, respectively; ``filtered'' retains only selections with $\hat\pi_j > 0.9$. The selection probabilities identify the false selections: filtering at $\hat\pi_j > 0.9$ reduces them from roughly $8$--$12$ per query point to fewer than $2.5$, and retains $86$--$95\%$ of the truly active features.}
\label{tab:sim-stability}
\end{table}

\bibliographystyle{abbrvnat}
\bibliography{main}
\end{document}